\documentclass[journal,12pt,draftclsnofoot,onecolumn]{IEEEtran}
\usepackage{cite}
\ifCLASSINFOpdf
\fi

\usepackage{amsmath,amsfonts,amsthm} 
\usepackage{bbm}
\usepackage{float}
\usepackage{dblfloatfix}
\usepackage{graphicx}
\usepackage{algorithm}
\usepackage{algpseudocode}
\usepackage{subcaption}
\usepackage{siunitx}
\usepackage{tikz}
\usepackage{pgfplots}
\usepackage[siunitx]{circuitikz}
\usepackage{epstopdf}
\usepackage{multirow}
\usepackage[font=scriptsize,skip=0pt]{caption}
\usepackage{soul}
\usepackage{color, colortbl}
\usepackage{array}
\usepackage{balance}

\usepackage{enumitem}
\usepackage{footmisc}
\usepackage{hyperref}
\usetikzlibrary{shapes.geometric}
\newcommand{\comment}[1]{}
\newcommand{\argmax}{\arg\!\max}

\newtheorem{theorem}{Theorem}

\newtheorem{lemma}[theorem]{Lemma}

\definecolor{Gray}{gray}{0.9}
\definecolor{LightCyan}{rgb}{0.88,1,1}

\begin{document}
	\title{Energy Harvesting Communication Using Finite-Capacity Batteries with Internal Resistance  \thanks{Rajshekhar Vishweshwar Bhat, Mehul Motani and Teng Joon Lim are with the Department of Electrical and Computer Engineering,  National University of Singapore, Singapore 117583.  Part of this work has been presented at the IEEE GLOBECOM 2015 conference in San Diego, CA, USA, 6-10 December, 2015 \cite{Globecom}.}}	
	\author{Rajshekhar~Vishweshwar~Bhat,~\IEEEmembership{Graduate~Student~Member,~IEEE,}
		Mehul~Motani,~\IEEEmembership{Senior~Member,~IEEE,}
		and~Teng~Joon~Lim,~\IEEEmembership{Senior~Member,~IEEE}}
	\maketitle
	\thispagestyle{plain}
	\pagestyle{plain}
	\vspace{-1.5cm}
\begin{abstract}
Modern systems will increasingly rely on energy harvested from their environment.
Such systems utilize batteries to smoothen out the random fluctuations in harvested energy. These fluctuations induce highly variable battery charge and discharge rates, which affect the efficiencies of practical batteries that typically have non-zero internal resistances. In this paper, we study an energy harvesting communication system using a finite battery with non-zero internal resistance.
We adopt a dual-path architecture, in which harvested energy can be directly used, or stored and then used. In a frame, both time and power can be split between energy storage and data transmission. For a single frame, we derive an analytical expression for the rate optimal time and power splitting ratios between harvesting energy and transmitting data. We then optimize the time and power splitting ratios for a group of frames, assuming non-causal knowledge of harvested power and fading channel gains, by giving an approximate solution. 
{When only the statistics of the energy arrivals and channel gains are known, we derive a dynamic programming based policy and, propose three sub-optimal policies, which are shown to perform competitively.} In summary, our study suggests that battery internal resistance significantly impacts the design and performance of energy harvesting communication systems and must be taken into account.


\end{abstract}
\IEEEpeerreviewmaketitle

\section{Introduction} \label{intro}
Natural energy harvesting (EH) promises near-perpetual operation of electronic devices due to its renewable nature. But, it poses several challenges in system design as the power generated from EH sources varies randomly with time, unlike conventional sources. For example, solar power can vary from \SI{1}{\micro\watt} to \SI{100}{\milli\watt} in a \emph{small-sized} (approximate  area of $\SI{10}{\square\centi\meter}$) solar cell across a day \cite{Kansal,non_linear}. The harvested energy needs to be stored in storage elements\footnote{We use `batteries' to refer to storage elements in general.} such as batteries and super-capacitors, for reliable system operation. In the process, due to source power fluctuations, the batteries are subjected to variable charging powers (rates). 
{In addition, it may be  required to vary the battery discharge powers (rates), for instance, to drain the battery quickly to accommodate the incoming harvested energy and, perhaps, also to cater to the variable power demand at the load (the wireless transmitter, in our case). Hence, in EH systems, both the charge and discharge powers are more variable and unpredictable than in conventional systems.} This necessitates a fundamental  change in the  way we store and use the harvested energy mainly given that the battery charge/discharge efficiencies (precisely defined later) depend on the charge and discharge powers \cite{ieee_ragone,Krieger}. This dependency  can be easily seen by considering a simple battery model - a voltage source/sink
with a series resistance.  Drawing a larger power from the battery entails a larger current, and hence a larger power loss in the internal resistance. Therefore charge and discharge efficiencies decrease with increasing charge and discharge powers, respectively \cite{Ragone,Krieger,ieee_ragone}.

In this work,  we consider a low-power wireless transmitter powered entirely by an EH source that is equipped with a battery having  capacity constraints with a non-zero internal resistance. The internal resistances of  commercial rechargeable micro-batteries and ultra-capacitors lie in the range of a few micro ohms to several tens of ohms \cite{Varta,Maxwell,EDLC}. Typically, for  small-sized wireless nodes, the harvested power lies in the range \SI{1}{\micro\watt}-\SI{100}{\milli\watt} and the discharge power can vary from \SI{10}{\micro\watt} to a few hundred milliwatts. In these ranges, by considering the simple  battery model  presented in \cite{Krieger}, it can be easily shown that the  charge efficiency range can be up to 15 percentage points and the discharge efficiency range can be as high as 30 percentage points.

{In this paper, we focus on applications that require the node to communicate $N_s$ channel symbols per frame over a fading channel.} For instance, a sensor network deployed in an Internet of Things (IoT) application consists of sensor nodes with limited data processing and storage capabilities \cite{data}. {These may be designed to deliver a fixed number of coded symbols per frame, due to limited data storage capacity at the receiver. The number of bits of information that are reliably transmitted within a frame can be varied by varying the information rate.} 
The harvested power in such applications can be very small due to limitations on the harvester size and area. 
{To illustrate the power management issues involved in such EH-based nodes, suppose for simplicity that the initial energy stored in the battery is zero. In this case, whenever the harvested power is lower than the power required for system operation, one cannot run the system from the EH source alone. We must first store the harvested energy in a battery and then, simultaneously draw power from the battery and the EH source, and run the system from the combined power. In such a scenario, it is sensible to ask how to optimally divide a frame into two parts -- the first to store energy in the battery, and the next to discharge energy from the battery for data transmission.}
Further, when the harvested power is \emph{high}, directing all the power to the battery may result in significant losses across internal resistances. 
In such cases, it may be beneficial to charge the battery  with only a  fraction of the harvested power while the transmission is carried out with the remaining power. In the second part of the frame, energy from the EH source may be directed to the load at the same time as energy from the battery, perhaps because neither the EH source nor the battery are able to power the load on their own. {In this paper, we develop novel policies for managing the battery charging and discharging schedules in such an EH-based transmitter. }

The problem of EH communications has been addressed from a variety of other perspectives as well.  A comprehensive  review of recent advances in  energy harvesting communications is presented in \cite{Review,survey}.  The information capacity of EH systems with infinite capacity batteries is derived in \cite{Ozel,vsharma_capacity} and  \cite{Tutu,Jog,Dong, Biplab} present the EH communication with finite batteries.   Other battery limitations such as,  leakage \cite{leakage,Leakage2}, non-linear charging  \cite{non_linear} and inefficiency \cite{shixin,inefficiency} have also been considered.  
The optimal policies when the system operation cost is zero and non-zero are studied in  \cite{DWF,Rui_DWF,chandraiisc, DGP,Rui_TI,Rui_CET}. 
We note that the authors of the current paper considered a similar EH communication problem in \cite{Globecom}.  The current work significantly extends the model in \cite{Globecom} by fully incorporating the effects of the battery internal resistance and providing more in-depth analysis. 

{The main contributions of this paper are as follows:
\begin{itemize}[leftmargin=*]
    \item  We identify generic and tractable models for the battery charge and discharge efficiencies which account for their dependency  on charging and discharging rates. This incorporates the effects of battery internal resistance. 
	\item  We then formulate a single frame optimization problem and  derive compact expressions for optimal time and power sharing ratios. 
	\item Further, we formulate an off-line optimization problem which assumes \emph{a priori} knowledge of the harvested powers and channel gains to obtain optimal time and power sharing ratios in the multiple frame case. We show that in general, the problem is a non-convex optimization problem, and propose an iterative algorithm to solve the problem approximately.  
	\item Further, assuming statistical knowledge and causal information of the harvested power and channel power gain variations, we solve for the optimal \emph{on-line} policy by using stochastic dynamic programming. We then propose three sub-optimal on-line algorithms which are practically feasible. Among them, an algorithm that is inspired by the approximate off-line solution achieves a significantly better performance compared to the other two algorithms.  
	\item We also show via numerical simulations that the optimal policy designed for an ideal battery performs poorly when the internal resistance is not negligible.
\end{itemize} }
The remainder of the paper is organized as follows. The system model and assumptions are presented in Section \ref{model}. Section \ref{single frame} and Section \ref{optimal}
address the single and multiple frame rate maximization problems respectively. 
Numerical results are presented in Section \ref{numerical_results} followed by concluding remarks in
Section \ref{conclusion}.

\section{System Model and Assumptions} \label{model}
\subsection{Block Diagram and System Operation}\label{sec:diagram_frame}
The block diagram of the system is given in Fig.  \ref{fig:BD}.  The principal components of the system are the power splitter, battery, power combiner and the transmitter.  The power splitter divides the instantaneous harvested power to simultaneously charge the battery and power the transmitter directly through a zero loss direct path.  The power combiner combines the power drawn from the battery and the direct path. The transmitter consumes $p$ W for circuit operation  during transmission but does not consume any power when not transmitting, as in \cite{Rui_CET}.   
       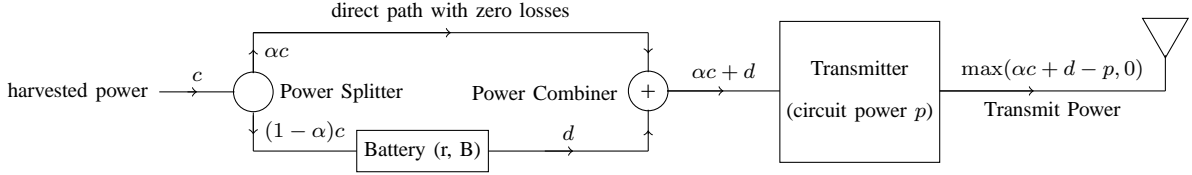
\begin{figure*}[t]
       	\centering
       	\begin{tikzpicture} [scale=2.5]
       	\draw  (0.85,0.125) -- (1.25,0.125) node [midway,above] {\scriptsize $c$}  node [left,xshift=-1cm] {\scriptsize harvested power};
       	\draw [->] (1,0.125) -- (1.05,0.125);
       	
       	\draw  (1.35,0.125) circle (3pt) node [align=right]  at (1.82,0.12)  {\scriptsize Power Splitter};
       	\draw  (1.35,0.225) -- (1.35,0.44) node[right,midway] {\scriptsize $\alpha c$};
       	\draw  (1.35,0.015) -- (1.35,-0.19) node[right,midway] {\scriptsize $(1-\alpha)c$};
       	\draw [->] (1.35,0.225) -- (1.35,0.35);
       	\draw [->] (1.35,0) -- (1.35,-0.1);
       	
       	\draw (1.35,0.44) -- (3.45,.44) node[above,midway] {\scriptsize direct path with zero losses};
       	\draw [->] (1.35,0.44) -- (2.4,0.44);
       	\draw (1.35,-0.19) -- (1.9,-0.19);
       	\draw (1.9,-0.30) rectangle (2.6,-0.08) node [align=center]  at (2.25,-0.19)  {\scriptsize Battery (r, B)};
       	\draw (2.6,-0.19) -- (3.45,-0.19) node [midway,above] {\scriptsize $d$};
       	\draw [->](2.9,-0.19) -- (3.05,-0.19);
       	\draw  (3.45,0.13) circle (3pt) node {\small +} node at (2.9,0.12) {\scriptsize Power Combiner};
       	\draw (3.45,-0.19) -- (3.45,.03);
       	\draw [->] (3.45,-0.1) -- (3.45,-0.08);
       	\draw (3.45,0.44) -- (3.45,.23);	
       	\draw [->] (3.45,0.35) -- (3.45,.33);
       	\draw (3.55,0.13) -- (4.15,0.13) node [midway,above] {\scriptsize $\alpha c+d$};
       	\draw [->](3.55,0.13) -- (3.85,0.13);
       	\draw (4.15,-0.25) rectangle (5,.5) node [align=center] at (4.575,0.25)  {\scriptsize Transmitter} node  [align=center] at (4.575,0.02)  {\scriptsize (circuit power $p$)};
       	\draw (5,.13) -- (6.2,0.13) node [midway,above] {\scriptsize $\max(\alpha c+d-p,0)$} node [midway,below] {\scriptsize Transmit Power}; 
       	\draw [->] (5,.13) -- (5.5,0.13);
       	\tikzset{
       		buffer/.style={
       			draw,
       			shape border rotate=-90,
       			isosceles triangle,
       			isosceles triangle apex angle=60,
       			node distance=.03cm,
       			minimum height=.03em
       		}
       	}
       	\draw (6.2,0.13) -- (6.2,.32);
       	\node at (6.2,.46)[buffer]{};
       	
       	\end{tikzpicture}
       	\caption{\scriptsize The dual-path EH communication system. A fraction ($0\leq\alpha\leq 1$) of the harvested power ($c$) can be directed to the load through the direct path. The remaining power is directed to the battery having capacity of $B$ joules and internal resistance of $r$ ohms. The battery is discharged at $d$ W.  The transmitter consumes $p$ W for its operation  during transmission but does not consume any power when not transmitting. }
       	\label{fig:BD} 
       	\vspace{-.75cm}
       \end{figure*}
The structure of the communication frame adopted in this work is shown in Fig. \ref{fig:Frame}. The harvested power, denoted by $c$, and channel power gain, denoted by $h$, are assumed to remain constant over the frame of length $\tau$ seconds. 
{We also assume that the channel bandwidth is $W\; \si{\hertz}$.}

We assume the battery cannot be charged and discharged simultaneously. This assumption is both practical and without loss of generality.  
From the practical perspective, charging and discharging of a battery/capacitor involves the movement of ions/electrons in mutually opposite directions and the particles can move in only one net direction at a time \cite{handbook}.
Mathematically one can relax the assumption and prove that charging and discharging a battery simultaneously is always suboptimal, similar to the arguments in \cite{inefficiency}. 

 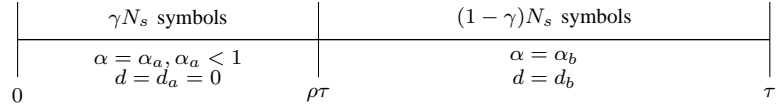
\begin{figure*}[t]
	\centering
	\begin{tikzpicture} [scale=10]
	\draw (0,0) -- (.4,0) node [below, midway] {\scriptsize $\alpha=\alpha_a, \alpha_a<1$}	;
	\draw (0,0) -- (.4,0) node [below, midway, yshift=-0.25cm] {\scriptsize $d=d_a=0$};	
	
	\draw (.4,0) -- (1,0) node [below, midway] {\scriptsize $\alpha=\alpha_b$}	;
	\draw (.4,0) -- (1,0) node [below, midway, yshift=-0.25cm] {\scriptsize $d=d_b$};

	\draw (0,0.05) -- (0, -0.05) node [below] {\scriptsize $0$};	;
	\draw (1,0.05) -- (1, -0.05) node [below] {\scriptsize $ \tau$};	
	\draw (.4,0.05) -- (.4, -0.05) node [below] {\scriptsize $\rho \tau$};	
	
	\draw (0,0) -- (.4,0) node [above, midway] {\scriptsize $\gamma N_s$ symbols};
	\draw (.4,0) -- (1,0) node [above, midway] {\scriptsize $(1-\gamma) N_s$ symbols};	
	\end{tikzpicture}
	\caption{\scriptsize The communication frame structure adopted in the paper. The frame length is $\tau$ seconds.  During $[0,\rho\tau)$, the power splitting ratio $\alpha=\alpha_a$, i.e., the power supplied to the transmitter is $\alpha_a c$ W. Over this time period, the battery must be charged, i.e., $\alpha_a<1$ and the discharge power is zero, i.e.,  $d=d_a=0$ W.  During $[\rho\tau,\tau]$, information must be transmitted, i.e., $(1-\gamma)N_s, \; 0< 1-\gamma\leq 1,$ symbols are transmitted. The power splitting ratio $\alpha=\alpha_b$, and the battery is charged at $(1-\alpha_b)c$ W and discharged at $d_b$ W, with $(1-\alpha_b)d_b = 0$, i.e., the battery cannot be charged and discharged at the same time. We assume that $\gamma N_s, \; 0\leq \gamma< 1,$ symbols are transmitted  in the first part of the frame. }
	\label{fig:Frame}
	\vspace{-.75cm}
\end{figure*}
{We assume an infinite backlog of data at the transmitter.
Based on the motivation provided in the introduction, to deal with the situation when the total available energy (the sum of the initial energy stored in the battery and the harvested energy) in a frame is lower than the total energy required to operate the system over the entire frame duration, we divide a communication frame into two phases, a charging phase in which the battery must be charged, and a transmitting phase, in which information must be transmitted.}
The frame structure (See Fig. \ref{fig:Frame}) is as follows:
\begin{itemize}[leftmargin=*]
	\item Over the time duration $[0,\rho\tau)$, $\rho \in [0,1]$, the battery must be charged, i.e., the charging rate is $(1-\alpha_a)c$ W with $0\leq\alpha_a<1$.
	 Since the battery cannot be charged and discharged simultaneously, the discharge power must be zero, i.e., $d=d_a=0$ W in this time duration. 	We assume that $\gamma N_s, \; 0\leq \gamma< 1,$ symbols   are transmitted by  utilizing the remaining $\alpha_a$ fraction of the harvested power from the direct path. 
	\item Over the time duration $[\rho\tau,\tau]$, information must be transmitted, i.e., $(1-\gamma)N_s, \; 0< 1-\gamma\leq 1,$ symbols are transmitted. The battery is charged at $(1 - \alpha_b)$ fraction of the harvested power or discharged at $d=d_b$ W. Whether the battery is being charged or discharged,  $\alpha_b$ fraction  of the harvested power is directly delivered to the transmitter.   
\end{itemize}
  The variable $\rho$, referred to as {\em time-splitting ratio} (TSR), is the ratio of the length of the charging phase to the total frame duration. The variable $\alpha$, referred to as  {\em power-splitting ratio} (PSR), indicates the  fraction of the harvested power directly used to power the transmitter. Note that by definition, in the first part of the frame the battery \emph{must be charged}, i.e., $\alpha_a<1$, but information may or may not be transmitted $(\gamma\geq 0)$.  However, in the second part of the frame, information \emph{must be transmitted}, i.e., $\gamma < 1$, but the battery may or may not be charged.

%
\begin{figure}[t]
	\centering
	\begin{tabular}{ll}
		\begin{subfigure}{0.5\textwidth}
			\centering
			\begin{tikzpicture}[scale=1.5]
			\filldraw [gray!5] (0.45,-0.85) rectangle (2.8,.95);
			\draw (0,-0.5) to[generic] (0,0.5) -- (0, 0.75) -- (0.5,0.75) to 
			(0.5,0.75) to[resistor,o-,label=\mbox{\scriptsize $r\; \si{\ohm}$}] (2,0.75) -- (2.5,0.75) -- (2.5,0.5) to
			(2.5,0.5) to[battery,label=\mbox{\scriptsize $V_B\; \si{\volt}$}] (2.5,-0.5) -- (2.5,-0.75)--(.5,-0.75) to[short,o-] (0,-0.75)--(0,-0.5);
			\draw [ ->] plot [smooth] coordinates { (0.3,.53) (1.9,.5) (2.1,0)} ;
			\node [] at (1.2, 0.35){\scriptsize Charge Current};
			\node [rotate=90] at (-0.3,0) {\scriptsize Power Source};
			\node at (2.6, 0.35) {\scriptsize +};
			\node at (2.6, -0.35) {\scriptsize -};
			\end{tikzpicture}
			\caption{An equivalent circuit diagram in the charge cycle. }
			\label{fig:charging circuit}
		\end{subfigure} 	
		& 
			\hskip-1cm
		\begin{subfigure}{0.5\textwidth}
			\centering
			\begin{tikzpicture}[scale=1.5]
			\filldraw [gray!5] (-0.3,-0.85) rectangle (2.05,.95);
			\draw (0,0.5) to[battery] (0,-0.5) -- (0, -0.75) -- (2,-0.75) to[short,o-] (2.5,-0.75)--(2.5,-0.5)
			(2.5,-0.5) to[generic] (2.5,0.5) -- (2.5,0.75)--(2,0.75) to[resistor,o-] (0.5,0.75) -- (0,0.75) -- (0,0.5);
			\draw [ ->] plot [smooth] coordinates { (0.3,.53) (1.9,.5) (2.1,0)} ;
			\node  at (1.2, 0.35){\scriptsize Discharge Current};
			\node [rotate=90] at (2.8,0) {\scriptsize Load};
			\node at (-0.7,0) {\scriptsize $V_B\; \si{\volt}$};
			\node at (1.25,1.1){\scriptsize $r\; \si{\ohm}$};
			\node at (-0.1, 0.35) {\scriptsize +};
			\node at (-0.1, -0.35) {\scriptsize -};
			\end{tikzpicture}
			\caption{An equivalent circuit diagram in the discharge cycle. }
			\label{fig:discharging circuit}
		\end{subfigure}
		
		\\
		\begin{subfigure}{0.5\textwidth}
			\centering
			\begin{tikzpicture}[scale=1.5]
			\draw [<->,thick] (0,2) node (yaxis) [above] {\scriptsize  $\mathcal{N}_c(c_p)$}
			|- (3,0) node (xaxis) [below,xshift=0.5cm] {\scriptsize charge power , $c_p$};
			\draw[yscale=2.5,domain=0:1.1,smooth,variable=\x] plot ({\x},{ln(1+(1.5-\x)/(1.5+\x))});
			\draw[yscale=2.5,domain=0:2.1,smooth,variable=\x] plot ({\x},{ln(1+(2.5-\x)/(2.5+\x))});
			\draw[dotted, yscale=2.5] (1.1,{ln(1+(1.5-1.1)/(1.5+1.1))})--(1.1,0) node [below] {\scriptsize $C_p(r_2)$};
			\draw[dotted, yscale=2.5] (2.1,{ln(1+(2.5-2.1)/(2.5+2.1))})--(2.1,0) node [below] {\scriptsize $C_p(r_1)$};
			\draw  (0,0) -- (0,1.73) node [left] {\scriptsize $\mathcal{N}_{c_0}$};    
			\draw [->] (2,1.8) -- (1,.9);
			\draw [->] (2,1.5) -- (1,.45);	
			\draw (1,.9) -- (2,1.8) node [right] {\scriptsize $r=r_1$};
			\draw (1,.45) -- (2,1.5) node [right] {\scriptsize $r=r_2>r_1$} ;				
			\end{tikzpicture}
			\caption{ The charging efficiency  model based on \cite{Krieger}. }
			\label{fig:CEDEa}
		\end{subfigure} 	
		& 

		\begin{subfigure}{0.5\textwidth}
			\centering
			\begin{tikzpicture}[scale=1.5]
			\draw [<->,thick] (0,2) node (yaxis) [above] {\scriptsize  $\mathcal{N}_d(d_p)$}
			|- (3,0) node (xaxis) [below] {\scriptsize discharge power, $d_p$};
			\draw[yscale=.8,xscale=1/4,domain=0:1.5*3,smooth,variable=\x] plot ({\x},{ln(1+(1.5*4-\x))}); 
			\draw[yscale=.8,xscale=1/4,domain=0:2.5*3,smooth,variable=\x] plot ({\x},{ln(1+(1.5*4-(1.5/2.5)*\x))}) ;
			\draw [dotted] (1.125,.7328) -- (1.125,0)  node [below] {\scriptsize $D_p(r_2)$};
			\draw [dotted] (1.875,.7328) -- (1.875,0)  node [below] {\scriptsize $D_p(r_1)$};
			\draw (0,0) -- (0,1.65) node [left] {\scriptsize $\mathcal{N}_{d_0}$};    
			\draw [->] (2,1.8) -- (1,1.25);
			\draw [->] (2,1.5) -- (1,.9);	
			\draw (1,1.25) -- (2,1.8) node [right] {\scriptsize $r=r_1$};
			\draw (1,.9) -- (2,1.5) node [right] {\scriptsize $r=r_2>r_1$} ;	            
			\end{tikzpicture}
			\caption{ The discharging efficiency  model based on \cite{Ragone}. }
			\label{fig:CEDEb}
		\end{subfigure}
	\end{tabular} 
	\caption{Equivalent circuit diagrams and charging/discharging efficiencies of a battery. The battery is modeled as a constant voltage source/sink with nominal voltage $V_B$ \si{\volt} with a series  internal resistance of $r$ \si{\ohm}.}
	\label{fig:CEDE}
	\vspace{-.8cm}
\end{figure}
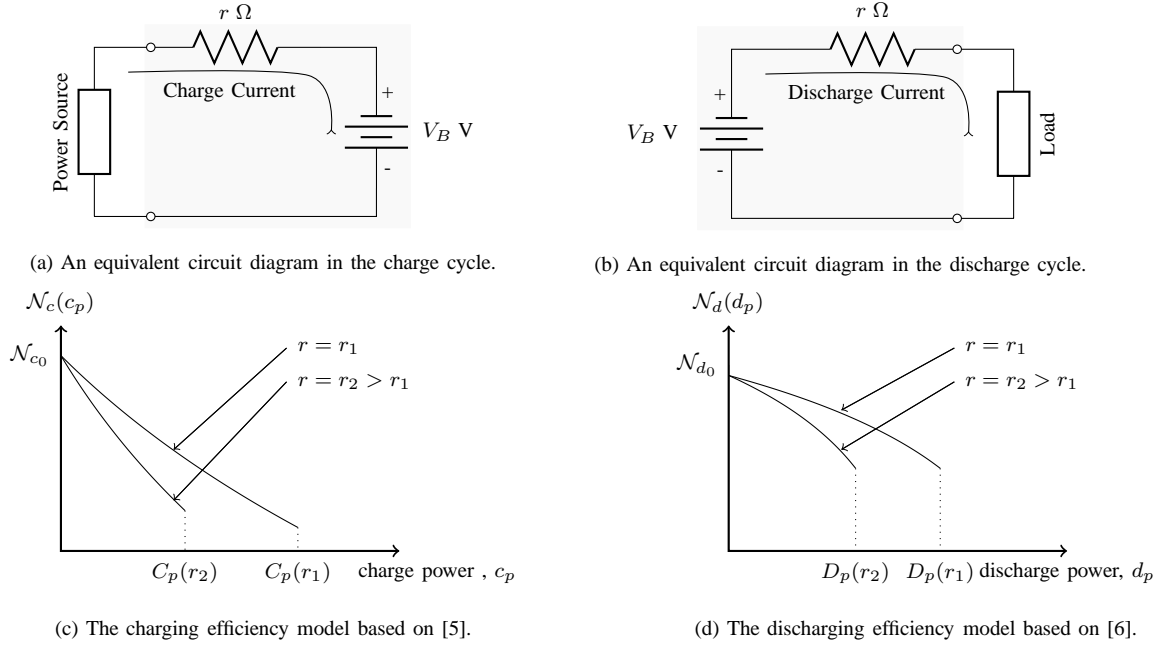

\subsection{Battery Charge/Discharge Efficiency Model}\label{sec:battery_model}
We assume that the battery capacity is $B$ joules and that it has a constant, non-zero internal resistance, denoted by $r$ ohms. We model a battery as an ideal voltage source/sink with a series  internal resistance (See Fig. \ref{fig:charging circuit} and Fig. \ref{fig:discharging circuit}). The losses across the internal resistance lead to battery inefficiencies.  

The charging efficiency of a battery, $\mathcal{N}_c(c_p,r)$ is defined as the ratio of the rate at which energy is stored in the battery, internally,  to the external charge power $c_p$. Based on \cite{Krieger}, we note that the charging efficiency is a convex decreasing function of the charge power, $c_p$, for a given $r$. This property is illustrated in Fig. \ref{fig:CEDEa}, where $\mathcal{N}_{c_0}$ is the vertical intercept and $C_p(r)$ is the maximum charge power constraint.
 Further, the internal charging power, $\mathcal{N}_c(c_p,r)c_p$,  is a concave function of $c_p$. 

The discharging efficiency, $\mathcal{N}_d(d_p,r)$ is defined as the ratio of the power delivered to the load, $d_p$, to the rate at which energy is drawn from the battery, internally. 
 It is shown to be a concave decreasing function of the external discharge power ($d_p$), for a given $r$ in \cite{Ragone}.  This property is illustrated in Fig. \ref{fig:CEDEb}, where $\mathcal{N}_{d_0}$ is the vertical intercept and, $D_p$, a concave decreasing function of $r$, is the maximum discharge power. Further, the internal discharging power, $d_p/\mathcal{N}_d(d_p,r)$,  is a convex function of $d_p$. We also note that it is physically impossible to charge and discharge a battery beyond $C_p$ and $D_p$, respectively. In the rest of the paper, we denote $\mathcal{N}_c(c_p,r)$ as $\mathcal{N}_c(c_p)$ and $\mathcal{N}_d(d_p,r)$ as $\mathcal{N}_d(d_p)$ for brevity. 

In general, the capacity of a battery varies with charge/discharge rates and this effect is referred to as the rate-capacity effect. 
In many cases, the rate-capacity effect can be easily mitigated with  additional circuitry \cite{TI} and, by avoiding battery overcharging or undercharging leading to extreme conditions \cite{Krieger}.  Further, \cite{Kansal_eff_Dual} argues that the rate-capacity effect is insignificant at  low power levels.  Hence, we do not account for the rate-capacity effect in this work.  
\section{Single-Frame Rate Optimization}\label{single frame}
For transmission over an additive white Gaussian noise (AWGN) channel with power gain $h$, transmit symbol energy $P$ and unit received noise power spectral density, the maximum achievable rate is $\log\left(1+hP\right)$ bits per channel symbol \cite{Cover06}.   
As in \cite{DWF,inefficiency}, we  assume that the channel power gain for the current frame remains constant and its value is known at the start of the frame.  {We assume that the number of coded (i.e. channel) symbols to be transmitted in a frame is fixed at $N_s$. }

For any given $\rho$,  the average rates within the two disjoint periods can be obtained as follows.
\paragraph{For $t\in[0,\rho\tau)$} Without loss of generality, assume that we transmit $\gamma N_s, \; 0\leq \gamma< 1$, symbols within the first part of the frame.  Since the transmitter is supplied with $\alpha_a c$ W (recall that $d_a=0$)  for $\rho\tau$ seconds directly from the EH source, the average symbol power $P_a=(\alpha_a c-p)\rho \tau/(\gamma N_s)$. If $\gamma=0$, then $P_a=0$.  Consequently, the  information rate is $R_a=\log(1+hP_a)$.
Note that we can transmit symbols only if $P_a>0$ implying that $(\alpha_a c-p)\rho \tau$ must be strictly greater than zero for the symbol transmission to take place. 
Hence, we have $\gamma=0$ if  $(\alpha_a c-p)\rho \tau<0$. 
Since the battery is charged at $(1-\alpha_a)c$ W, the amount of energy stored in the battery over $[0,\rho\tau)$ is $B_{\rho\tau}=\mathcal{N}_c((1-\alpha_a)c)(1-\alpha_a)c\rho\tau$.
 
\paragraph{For $t\in[\rho\tau,\tau]$} In the second  part of the frame, the EH source and the battery  supply  $\alpha_b c$ W  and $d_b$ W, respectively, to the transmitter, with $(1-\alpha_b)d_b = 0$ as the battery cannot be charged and discharged at the same time. Since the number of transmitted symbols is $(1-\gamma)N_s$, the average symbol power, $P_b =(\alpha_b c-p+d_b)(1-\rho)\tau/((1-\gamma) N_s)$ and the information rate is $R_b=\log(1+hP_b)$.   
Since the battery charging power over this time period is $(1-\alpha_b)c$, internally the harvested energy gets stored in the battery at the rate of $\tilde{c_b}=\mathcal{N}_c((1-\alpha_b)c)(1-\alpha_b)c$ W. Since $d_b$ is the discharge power, internally the battery energy gets drawn at $\tilde{d_b}=d_b/\mathcal{N}_d(d_b)$ W. 

Consolidating the information rates within the above two disjoint periods,  the average rate in the frame is given by, 
\begin{align}\label{eq:Rateind}
R(\rho,\alpha_a,\alpha_b,\gamma,d_b)=\gamma R_a +(1-\gamma )R_b 
\end{align} 
{Before formulating the optimization problem, we make an important remark on the generality of the \emph{two-phased} frame structure described in Section \ref{sec:diagram_frame}. In the proposed frame structure, note that the charging and discharging rates can take at most two values in a frame as per values of $\alpha_a$, $\alpha_b$, $d_a$ and $d_b$. To understand why it is sufficient to divide the frame into two phases, consider a frame that is divided into more than two phases with possibly different charging, discharging and transmit powers in each of the phases. Now, note that the internal charging powers, discharging powers and information rates are concave, convex and concave functions of the external charging, discharging and transmit powers, respectively. Hence, the loss across the internal resistance is minimized and, simultaneously, the information rate is maximized when the battery is charged and discharged at uniform powers. Hence, we can always replace any number of phases with a single phase without any loss of optimality. As described earlier, it may not be feasible to have a frame with only one phase as the amount of energy required to run the system over the entire frame may be more than the amount of energy available. Hence, we conclude that the frame structure described in Section \ref{sec:diagram_frame} is completely general and sufficient to extract the maximum possible performance from the system.}

To maximize the information rate per frame, we must thus solve the following optimization problem:
{
\begin{subequations}\label{eq:singlerate}
	\begin{alignat}{2}
	\text{P0}:\quad \underset{\rho, \alpha_a, \alpha_b, \gamma, d_b}{\text{maximize}} \quad &    R(\rho,\alpha_a,\alpha_b,\gamma,d_b) \quad \quad   \\  
	\textrm{subject to} \quad &  (\tilde{d_b}-\tilde{c_b})(1-\rho)\tau-B_{\rho\tau}-B_0 \leq 0  \label{eq:sc1} \\
	&   B_0+B_{\rho\tau}-(\tilde{d_b}-\tilde{c_b})(1-\rho)\tau-B\leq 0   \label{eq:sc2} \\
	&   0\leq \rho \leq 1\label{eq:sc3a}\\
	&   \alpha_c\leq \alpha_a, \alpha_b\leq 1\label{eq:sc3b}  \\
	&   0\leq  d_b \leq D_p  \label{eq:sc4} \\
	&   (1-\alpha_b)d_b=0    \label{eq:s15}
	\end{alignat}
\end{subequations}
where $\alpha_c=1-C_p/c$ and \eqref{eq:sc1} is the energy causality constraint which says that energy drawn from the battery ($\tilde{d_b}(1-\rho)\tau$) has to be less than or equal to the energy stored in the battery ($\tilde{c_b}(1-\rho)\tau+B_{\rho\tau}+B_0$). The inequality in \eqref{eq:sc2} is the battery capacity constraint, \eqref{eq:sc3b} accounts for the maximum charge rate constraint, i.e., the charge rate $(1-\alpha)c$ must not exceed the maximum charge rate $C_p$, \eqref{eq:sc4} is the maximum discharge rate constraint and \eqref{eq:s15} captures the fact that the battery cannot be charged and discharged simultaneously. Recall that $\tilde{d_b}$ and $\tilde{c_b}$ are functions of $d_b$ and $\alpha_b$, respectively.}

We now make the following observation which says that it is not optimal to transmit any symbols in the first part of the frame in the optimal solution to P0.  
\begin{lemma}\label{lemma:transmitpower}
	In the optimal solution to {\normalfont P0} in \eqref{eq:singlerate},
	\begin{enumerate}
		\item the total number of symbols transmitted and the average rate over $[0,\rho^*\tau]$ are zero, i.e.,  $\gamma^* N_s=0$ and $R_a^*=0$ and, 
		\item all $N_s$ symbols are transmitted  during $[\rho^*\tau,\tau]$ at the average power $(\alpha_b^*c-p+d_b^*)(1-\rho^*)\tau/N_s$.
	\end{enumerate} 
\end{lemma}
\begin{proof}
	See Appendix A.
\end{proof}
As a result of the above lemma, the objective function of P0 in \eqref{eq:singlerate} can be rewritten as $R(\rho,\alpha_a,\alpha_b,\gamma,d_b)=(1-\gamma^*)R_b  =\log((\alpha_bc-p+d_b)(1-\rho)\tau/N_s)$. Note that the optimal value of $\gamma^* = 0$, i.e. $\gamma$ is no longer an optimization variable. {The result also highlights that the number of symbols transmitted in both the phases in the optimal case is always an integer, thus satisfying requirements of practical applications.}
Though the objective function now has a simpler form, due to coupling of $\rho$, $\alpha_a$, $\alpha_b$ and $d_b$, P0 is still a non-convex optimization problem.  
However, we exploit the structure of the problem and present the optimal solution in the following theorem. 
 {\begin{theorem}\label{thm:opt_save_ratio_gen}
The optimal solution to {\normalfont P0} is $\gamma^*=0$, $\alpha_a^*= \argmax_{\alpha_c\leq \alpha\leq 1}\left(\mathcal{N}_c((1-\alpha)c)(1-\alpha)c\right)$, $\alpha_b^*=1$,  $\rho^*=\min(\rho_B,\rho_r)$, 
		where  $\rho_r=\argmax_{\rho} \left((\alpha_b^* c-p+\hat{d}_b(\alpha_a^*,\rho))(1-\rho)\right)$ and $\rho_B=(B-B_0)/(\mathcal{N}_c(c_p^*)c_p^*\tau)$, and $d_b^*=\hat{d}_b(\alpha_a^*,\rho^*)$, where we define $c_p^*=(1-\alpha_a^*)c$, $\hat{d}_b(\alpha_a^*,\rho)=\{\min(d_b,D_p):d_b/\mathcal{N}_{d}(d_b)= (\mathcal{N}_c(c_p^*)c_p^*\rho+B_0/\tau)/(1-\rho)\}$. 
\end{theorem}
\begin{proof}
	See Appendix B.
\end{proof}
Theorem \ref{thm:opt_save_ratio_gen} says that in the optimal solution, the battery is charged at the optimal rate in the charging phase. Recall that no information is transmitted in the charging phase. In the transmitting phase, the information is transmitted with the power drawn from the battery and the EH source. 
 Since the optimal external charging rate of the battery, $(1-\alpha_a^*)c$, may be lower than the harvested power, $c$ and because the transmission is not carried out in the charging phase, the remaining $\alpha_a^*c$ \si{\watt} gets wasted.} 


So far we did not impose any constraint on the channel bandwidth. {Now, recall that the channel bandwidth is $W\; \si{\hertz}$} and note that we need to transmit $N_s$ symbols during $[\rho\tau,\tau]$, i.e., in $(1-\rho)\tau$ seconds.  Hence the Nyquist bandwidth\footnote{considering a raised cosine pulse shaping filter with unity roll-off factor.} is $N_s/(1-\rho)\tau$. 
Since,  the signal bandwidth has to be less than the channel bandwidth, we must have $N_s/(1-\rho)\tau \leq W$, i.e., 
\begin{align}\label{eq:BW}
\rho \leq 1-\frac{N_s}{W\tau}=\rho_{W}
\end{align}   Hence, the optimal TSR  with the bandwidth constraint is $\rho^*_{\mathrm{BW}}=\min(\rho^*,\rho_W)$, where $\rho^*$ is obtained from Theorem  \ref{thm:opt_save_ratio_gen}.

\section{Multiple Frame Average Rate Optimization } \label{optimal}
{In this section, we consider the problem of average rate maximization across multiple communication frames with the number of frames denoted by $N$.} 
The harvested power in any frame $i$ is assumed to be a random variable $C_i$ with a finite support, i.e., $0 \leq C_i <\infty, \, i=1,\ldots,N$.  We assume that the random variables, $C_1$, $C_2$,$\dots$, $C_{N}$, are independent and identically distributed and that they do not change within a frame. 
Further, we assume that the channel power gains, denoted by $H_1$, $H_2$,$\dots$, $H_{N}$, in frames $1, 2, \ldots,N$, respectively, are independent and identically distributed. The frame duration is assumed to be $\tau$ for all the frames.  
\subsection{Problem Formulation} \label{offline}
First, we consider the off-line optimization under the assumption that the harvested power  and channel gains are \emph{a priori} known at the transmitter as in \cite{DWF,DGP,inefficiency,Rui_TI, Rui_CET}.
The optimal throughput under the off-line optimization gives an upper bound for the optimal throughput in all on-line algorithms. 
Let the realizations of  harvested powers and channel power gains in frames, $1,\ldots,N$, be $c_1,\ldots, c_N$, and $h_1,\ldots,h_N$,  respectively.
The average throughput across $N$ frames is given by\footnote{ where any bold symbol $\mathbf{x}=\{x_1,\ldots,x_N\}$. }, 
\begin{align}\label{eq:avgRate}
R_{\mathrm{avg}}(\boldsymbol{\rho},\boldsymbol{\alpha}_a,\boldsymbol{\alpha}_b,\boldsymbol{\gamma},\mathbf{d}_b)=\frac{1}{N}\sum_{i=1}^{N}R(\rho_i, \alpha_{a_i},\alpha_{b_i}, \gamma_i,d_{b_i})
\end{align} 
where $R(.)$ is given by  \eqref{eq:Rateind} with $h=h_i$ for frame $i$. Note that all the variables carry their usual meanings except that they are now indexed by the frame indexes.
 
To maximize the average information rate across $N$ frames, we must thus solve the following optimization problem:
\begin{subequations}\label{eq:average_rate1}
	\begin{alignat}{2}
	\text{P1}:\quad   \underset{\boldsymbol{\rho},\boldsymbol{\alpha}_a,\boldsymbol{\alpha}_b,\boldsymbol{\gamma},\mathbf{d}_b}{\text{maximize}} \; & R_{\mathrm{avg}}(\boldsymbol{\rho},\boldsymbol{\alpha}_a,\boldsymbol{\alpha}_b,\boldsymbol{\gamma},\mathbf{d}_b)    \\  
 \textrm{subject to}\;	&  \sum_{k=1}^{i}\left(\tilde{d}_{b_k}(1-\rho_k) -\tilde{c}_{a_k}\rho_k-\tilde{c}_{b_k}(1-\rho_k)\right)\tau-B_0 \leq 0 \label{eq:gc11}  \\
	&  B_0+\sum_{k=1}^{i}\left(\tilde{c}_{a_k}\rho_k+\tilde{c}_{b_k}(1-\rho_k) -\tilde{d}_{b_k}(1-\rho_k)\right)\tau -B \leq 0 \label{eq:gc12} \\
	& {\alpha_{c_i} \leq \alpha_{a_i}, \alpha_{b_i}  \leq 1}, \;\; 0\leq \rho_i \leq \rho_{W}  \label{eq:sc5}\\
	& (1-\alpha_{b_i})d_{b_i}=0, \; 0\leq d_{b_i} \leq D_p \label{eq:sc6}
	\end{alignat}
\end{subequations}
for $i=1,\ldots,N$, where $ \alpha_{c_i}=1-C_p/c_i$ and,  $\tilde{c}_{a_k} =(1-\alpha_{a_k}) c_k\mathcal{N}_c((1-\alpha_{a_k})c_k)$, $\tilde{c}_{b_k} =(1-\alpha_{b_k}) c_k\mathcal{N}_c((1-\alpha_{b_k})c_k)$ are
concave functions in $\alpha_{a_k}$ and $\alpha_{b_k}$, respectively and, they  specify the internal charging power of the battery over the time durations $[0,\rho_k\tau)$ and $[\rho_k\tau,\tau]$ in any frame $k$, respectively. The internal discharge power $\tilde{d}_{b_k}=d_{b_k}/\mathcal{N}_{d}(d_{b_k})$ in any frame $k$ over $[\rho_k\tau,\tau]$ is  a convex function of $d_{b_k}$  and  $B_0$ is the initial energy stored in the battery. The constraints in \eqref{eq:gc11} and \eqref{eq:gc12} are energy causality and battery capacity constraints, respectively.  Note that we have also included the  bandwidth constraint, maximum charge and discharge rate constraints and the constraint that the simultaneous charging and discharging is infeasible in  \eqref{eq:sc5} and \eqref{eq:sc6}.
As in the single frame case, we note the following.
\begin{lemma} \label{lemma:gamma}
	In the optimal policy,  $\gamma_i^*=0$  for all $i=1,\ldots, N$. 
\end{lemma}
\begin{proof}
	See Appendix  C. 
\end{proof}
Hence, $\gamma_i$'s are no longer  optimization variables.  Hence, the transmit power $P_i$ in any frame $i$ is equal to $(\alpha_{b_i}c_i-p+d_{b_i})(1-\rho_i)\tau$. 
Clearly, P1 in \eqref{eq:average_rate1} is a non-convex optimization problem due to the non-convex constraint in \eqref{eq:gc12} and  due to the coupling of $\rho_i$'s with $\gamma_i$'s, $d_{b_i}$'s,  $\alpha_{a_i}$'s and $\alpha_{b_i}$'s. 

In the following, we first solve the problem when the circuit cost is zero  and get some interesting insights on the optimal solution. We then approximately solve the problem when the circuit cost is non-zero. 
\subsection{Zero Circuit Cost ($p=0$) Case }
Since energy is not expended for the circuit operation during the transmission, in this case, we can transmit the coded symbols for the entire frame duration. Hence, $\rho_i$'s and ${\alpha}_{a_i}$'s are no longer  optimization variables. In this case, the optimization problem P1 in \eqref{eq:average_rate1} can be reformulated as
 \begin{subequations}\label{eq:zp}
	\begin{alignat}{2}
	P2:\quad \underset{\substack{\alpha_{b_i}, d_{b_i} \\ i=1,\ldots, N}}{\text{minimize}} \quad &  -\frac{1}{N}\sum_{i=1}^{N}\log\left(1+h_i(\alpha_{b_i}c_i+d_{b_i})\tau/N_s\right)  \\   \textrm{subject to} \quad	& \quad\eqref{eq:gc11}, \eqref{eq:gc12}, 0\leq d_{b_i}\leq D_p,  \alpha_{c_i}\leq \alpha_{b_i}\leq 1  , \quad i=1,\ldots,N  \label{eq:p2c5} 
	\end{alignat}
\end{subequations}
where the constraints should be self-explanatory.
In general, P2 is non-convex due to the non-convex constraint \eqref{eq:gc12}. 

When the channel gains remain constant across the frames, i.e., $h_i=h, \; i=1,\ldots, N$ and when the battery capacity is infinite, we make an interesting observation in the following theorem. 
\begin{theorem}\label{thm:zeropower}
	Consider any two frames, $j$ and $k \;(k>j)$ such that the battery has a non-zero residual energy in and between the frames $j$ and $k$. Then, while the battery is being charged, i.e., $\alpha_{b_j}, \alpha_{b_k}<1$, or the battery is being discharged, i.e., $d_{b_j}, d_{b_k} >0$,  the optimal transmit power  is a strictly monotonically increasing function of the harvested power, i.e.,  $c_j<c_k$ implies $P_j<P_k$. 
\end{theorem}
\begin{proof}
See Appendix D. 
\end{proof}
With the assumption that the battery efficiencies are independent of  the charge and discharge rates, it has been shown in \cite{inefficiency} that the optimal power allocation has a double threshold structure: the optimal transmit power does not vary with the harvested power whenever the harvested power is above an upper threshold or below a lower threshold. It is interesting to note that if the battery efficiencies vary with the charge and discharge rates as a result of non-zero internal resistance, the optimal transmit power strictly monotonically increases with the harvested power and does not exhibit the simple threshold structure observed with the fixed efficiency model in \cite{inefficiency}.

\subsection{Non-Zero Circuit Cost ($p>0$) Case}
Recall that P1 in \eqref{eq:average_rate1} is non-convex when the circuit cost is non-zero. 
Hence, analytically solving  P1 is challenging.  In the rest of the section, we approximately solve P1 by considering an upper bounding function of the discharge efficiency curve in Fig. \ref{fig:CEDEb}.  We define the following bounding function which is referred to as the step discharge model: $\mathcal{N}_d(d_b)=\mathcal{N}_{d_0}$ if $d_b\leq D_p$; $\mathcal{N}_d(d_b)=0$ otherwise, where $D_p$ is the maximum discharge rate.  

We can now eliminate the coupling between $d_{b_i}$'s and $\rho_i$'s by substituting $e_{b_i}=d_{b_i}(1-\rho_i)\tau$.  The constraint on the discharge rate in the step discharge model can be re-written as,
\begin{align}\label{eq:e}
e_{b_i}=d_{b_i}(1-\rho_i)\tau\leq D_p(1-\rho_i)\tau,\; i=1,\ldots,N
\end{align}  
To eliminate the coupling between $\rho_i$'s and $\alpha_{a_i}$'s, we make the following observation.
\begin{lemma}\label{lemma:optALPHA}
	Let $\alpha_{a_i}^*=\argmax_{\alpha_{c_i}\leq \alpha\leq 1}\left( \mathcal{N}_c((1-\alpha)c_i)(1-\alpha)c_i\right)$. Then, $R_{\mathrm{avg}}(\boldsymbol{\alpha}_a)\leq R_{\mathrm{avg}}(\boldsymbol{\alpha_{a}^*})$ for any given $\{\rho_i,\alpha_{b_i},e_{b_i}\}_{i=1}^N$. 
\end{lemma}
\begin{proof}
See Appendix E. 
\end{proof}
Lemma \ref{lemma:optALPHA} implies that in the optimal solution to P1, we must have $\boldsymbol{\alpha}_a=\boldsymbol{\alpha}_a^*$ always. 
To eliminate the coupling between $\alpha_{b_i}$'s and $\rho_i$'s, we make the following observation  which says that it is not optimal to charge the battery in the second part of the frame whenever $\rho_i>0$. 
\begin{lemma}\label{lemma:rgtzero}
	In the optimal policy,  if  the optimal $\rho_i>0$, then the optimal $\alpha_{b_i}=1$ and if the optimal $\rho_i=0$, then $\alpha_{c_i}\leq\alpha_{b_i}\leq1$ in the optimal case.  
\end{lemma}
\begin{proof}
	See Appendix F.
\end{proof}
The above result implies that $(1-\alpha_{b_i})\rho_i=0$.  Hence, the rate in frame $i$ can be re-written as   $R(\rho_i, \alpha_{b_i},\alpha_{a_i}, e_{b_i})=\log\left(1+h_i((\alpha_{b_i}-\rho_i)c_i\tau-p(1-\rho_i)\tau+e_{b_i})/N_s\right)$, where we have substituted  $d_{b_i}(1-\rho_i)\tau$ by $e_{b_i}$  and  $\alpha_{b_i}(1-\rho_i)$ by $\alpha_{b_i}-\rho_i$.   
Hence,  P1 can be reformulated as,
\begin{subequations}\label{eq:approxA}
	\begin{alignat}{2}
	\text{P3}:  \underset{\substack{\rho_i, \alpha_{b_i}, e_{b_i} \\ i=1,\ldots, N}}{\text{minimize}} & \; -\frac{1}{N}\sum_{i=1}^{N}\log\left(1+h_i((\alpha_{b_i}-\rho_i)c_i\tau-p(1-\rho_i)\tau+e_{b_i})/N_s\right)  \\  
	\textrm{subject to}\;	&   \sum_{k=1}^{i}e_{b_k}/ \mathcal{N}_{d_0} - \rho_k\tilde{c}_{a_k}^*\tau-\tilde{c}_{b_k}(1-\rho_k)\tau-B_0 \leq 0  \label{eq:mc1} \\
	&  B_0+\sum_{k=1}^{i}(\rho_k\tilde{c}_{a_k}^* \tau+\tilde{c}_{b_k} (1-\rho_k)\tau -e_{b_k}/ \mathcal{N}_{d_0}-B \leq 0  \label{eq:mc2} \\
	&  0\leq \rho_i\leq  \rho_W                               \label{eq:c3}   \\
	&  0\leq  e_{b_i} \leq D_p(1-\rho_i) \tau                           \label{eq:c4}  \\
	& (1-\alpha_{b_i})\rho_i=0  , \; \alpha_{c_i}\leq \alpha_{b_i}\leq 1     \label{eq:c5}  
	\end{alignat}
\end{subequations}
 for $i=1,\ldots,N$, where $\tilde{c}_{a_k}^*= (1-\alpha_{a_k}^*) c_k\mathcal{N}_c((1-\alpha_{a_k}^*)c_k)$ and the  constraints \eqref{eq:BW}, \eqref{eq:e}, \eqref{eq:gc11} and  \eqref{eq:gc12} are re-written as \eqref{eq:c3}, \eqref{eq:c4}, \eqref{eq:mc1} and \eqref{eq:mc2},  respectively.

As a result of Lemma \ref{lemma:rgtzero}, we need to optimize only over $\alpha_{b_i}$ if $\rho_i=0$ and optimize only  over $\rho_i$ if $\rho_i> 0$, because the optimal $\alpha_{b_i}=1$ whenever $\rho_i> 0$. If we know whether $\rho_i=0$ or $\rho_i>0$ for any frame $i$ in the optimal solution, $\rho_i$ and $\alpha_{b_i}$ get decoupled and we can obtain the solution to P3 by solving the resulting convex optimization problem. {However, the challenge is to identify whether $\rho_i>0$ or $\rho_i=0$ for $i=1,\ldots,N$, as the size of the search space increases exponentially with the number of frames, $N$.} In the sequel, we give an approximate solution to P1 by approximately solving P3. 
To identify if $\rho_i>0$ or $\rho_i=0$ and solve P3 approximately,  we adopt the following technique. 

\begin{enumerate}[leftmargin=*]
\item  In order to eliminate the coupling between $\rho_i$'s and $\alpha_{b_i}$'s, set $\alpha_{b_i}=1$ for each $i=1,\ldots, N$, and solve the modified P3, which now is a convex optimization problem, to obtain the optimal solutions $\{\rho_{\{i,\alpha_{b, \{1\leq j\leq N\}}=1 \}}, e_{b,\{i,\alpha_{b, \{1\leq j\leq N\}}=1  \} }\}_{i=1}^N$. Let  $E_{\{i,\alpha_{b, \{1\leq j\leq N\}}=1  \} }$ and $R_{\{i,\alpha_{b, \{1\leq j\leq N\}}=1 \}}$  be the optimal transmit energy and rate in frame $i$, respectively. 
Let the total energy loss, i.e., sum of charging and circuit losses  to achieve rate $R_{\{i,\alpha_{b, \{1\leq j\leq N\}}=1  \}}$  in frame $i$ be  $L_{\{i,\alpha_{b, \{1\leq j\leq N\}}=1 \}}$.  
\item Then, we set $\rho_i=0$ and find $\tilde{\alpha}_{b_i}$ which results in the same transmit energy of $E_{\{i,\alpha_{b, \{1\leq j\leq N\}}=1  \} }$ as in Step (1) for each $i=1,\ldots,N$. Now, $\tilde{\alpha}_{b_i}$ may not be feasible due to the maximum charge rate constraint. Hence, we consider $\alpha_{b,\{i, \rho_i=0\}}=\max(\alpha_{c_i},\tilde{\alpha}_{b_i})$ where the term $\alpha_{c_i}$ accounts for the maximum charge rate constraint. Let the total loss incurred with PSR of $\alpha_{b,\{i, \rho_i=0\}}$ in frame $i$ be $L_{i,\rho_i=0}$.   For any frame $i$,  we set $\rho_i^*=0$  if $L_{i, \rho_i=0} \leq L_{\{i,\alpha_{b, \{1\leq j\leq N\}}=1  \}}$; set  $\alpha_{b_i}^*=1$ otherwise. In the previous step, due the assumption that the charging and the transmission are not done simultaneously, i.e., $\alpha_{b, \{1\leq j\leq N\}}=1$, the charging losses will be high in frames with \emph{high} harvested powers. In the current step, we try to reduce the loss while maintaining the transmit power same as the previous step.  
Further, we note that if any frame $i$ receives energy in Step (1) above, then $\alpha_{b_i}^*=1$ because the frame that receives energy in Step (1) must receive energy in any other policy that performs better than the performance of the policy in Step (1).   If a frame receives energy, it is not optimal to charge the battery while it is being discharged, hence, we set $\alpha_{b_i}^*=1$, if frame $i$ receives energy in Step (1). 
\item Suppose $e_{b_i}^*$ is the solution in the above steps, then we assign $d_{b_i}^*=\{\min(d_{b_i},D_p): d_{b_i}(1-\rho_i^*)\tau/\mathcal{N}_d(d_{b_i})=e_{b_i}^*\}$ for all the frames $i=1,\ldots,N$. 
\end{enumerate}

\begin{algorithm}[t]
	\caption{ {Proposed Algorithm for Approximately Solving P1}}
	\label{non-causal_gen}
	\begin{algorithmic}[1]
		\Procedure{energy-alloc}{$B_0$,\textbf{c, h, N}}
		\State Compute $\alpha_{a_i}^*=\argmax_{\alpha_{c_i}\leq \alpha\leq 1}\left( \mathcal{N}_c((1-\alpha)c_i)(1-\alpha)c_i\right)$ and assign $\boldsymbol{\alpha}_a=\boldsymbol{\alpha}_a^*$.
		\State Solve P3 with $\alpha_{b_1}=\ldots=\alpha_{b_N}=1$. Obtain transmit powers $E_{\{i,\alpha_{b, \{1\leq j\leq N\}}=1  \} }$.
		\For {$i: 1 \rightarrow N$ }
		\State If frame $i$ receives energy, then $\alpha_{b_i}^*=1$; $F \leftarrow i$.  
		\State Set $\rho_i=0$ and compute $\tilde{\alpha}_{b_i}$ that results in transmit energy equal to $E_{\{i,\alpha_{b, \{1\leq j\leq N\}}=1  \} }$.
		\State Obtain $\alpha_{b,\{i, \rho_i=0\}}=\max(\alpha_{c_i},\tilde{\alpha}_{b_i})$. 
		\State Compute the total loss  $L_{\{i, \alpha_{b,\{1\leq j\leq N\}}=1  \}}$ and $L_{i, \rho_i=0}$  if $ i\notin F$. 
		\If {$L_{i, \rho_i=0} \leq L_{\{i,\alpha_{b,\{1\leq j\leq N\}}=1  \}}$} \hspace{.2cm} $\rho_i^*=0$; $A \leftarrow i$
		\Else  \hspace{.2cm} $\alpha_{b_i}^*=1$; $B \leftarrow i$
		\EndIf
		\State Substitute $\alpha_{b_i}=1,\forall\; i\in B\cup F$ and  $\rho_i=0, \forall\; i\in A$.
		\EndFor
		\State Solve the resultant  convex optimization problem to obtain $\boldsymbol{\rho}^*$ and $\mathbf{e}_{b}^*$. 
		\State $d_{b_i}^*=\{\min(d_{b_i},D_p):d_{b_i}(1-\rho_i^*)\tau/\mathcal{N}_d(d_{b_i})=e_{b_i}^*\}$ for each $i\in\{1,\ldots,N\}$.  
		\State Return  $\boldsymbol{\rho}^*,\boldsymbol{\alpha}_a^*,\boldsymbol{\alpha}_b^*,\mathbf{d}_b^*$.
		\EndProcedure
	\end{algorithmic}
\end{algorithm}

We present the algorithm in Algorithm \ref{non-causal_gen}. The convergence of the algorithm is guaranteed as the average rate increases in each iteration. {The computational complexity analysis of Algorithm \ref{non-causal_gen} is given as follows. We first note that two separate convex optimization problems are solved in Step 3 and Step 14,  each with a worst-case polynomial complexity in $N$. The complexity of the remaining steps is linear in $N$. Hence, we conclude that the worst case complexity of Algorithm \ref{non-causal_gen} is polynomial in $N$. Specifically, the computational complexity of Algorithm \ref{non-causal_gen} is $\mathcal{O}(N^{3})$ when interior-point methods are used to solve the convex optimization problems \cite{complexity}.}

 {The rationale behind some important steps in the Algorithm \ref{non-causal_gen} are as follows. Based on Lemma \ref{lemma:optALPHA}, we note that Step 2 gives the optimal result. Based on Lemma \ref{lemma:rgtzero} and the preceding discussion (in points 1 to 3), we can see that the result obtained in Step 3 to Step 13 is close to the optimal result. Further, Step 14 gives the optimal result as the optimization problem being solved is convex.}
Table \ref{my-label} provides a summary of various optimization problems considered so far with some useful comments. 
\begin{table*}[t]
	\centering
	\begin{tabular}{|p{2cm}|p{6cm}|p{7cm}|}
		\hline
		\begin{center} Circuit Cost, $p$ \end{center} &\begin{center} Single Frame Case, P0 in \eqref{eq:singlerate}	\end{center} & \begin{center} Multiple Frame Case, P1 in \eqref{eq:average_rate1}	\end{center}  \\ \hline
		\begin{center}{$p=0$}  \end{center} &{$\gamma^*=\rho^*=d_b^*=0$, $\alpha_b^*=1$, $\alpha_a$ does not play any role. Battery is neither charged nor discharged.} & {P1 is reformulated as P2 in \eqref{eq:zp}. Further, $\gamma_i^*=\rho_i^*=0$, $i=1,\ldots,N$ and $\alpha_{a_i}$'s do not play any role; optimization variables: $\alpha_{b_i}$'s and $d_{b_i}$'s. See Theorem \ref{thm:zeropower}.} \\ \hline
		\begin{center} {$p>0$}  \end{center} &{$\gamma^*=0$, either $\rho^*> 0$, $\alpha_a^*<1$, $\alpha_b^*=1$, $d_b^*>0$ or $\rho^*=d_b^*=0$, $\alpha_b^*=1$ with arbitrary $\alpha_a$. 
			The harvested energy is stored at the optimal rate over $[0,\rho^*\tau)$. See Theorem  \ref{thm:opt_save_ratio_gen}. } & {Using an approximate discharge efficiency model, P1 is reformulated as P3 in \eqref{eq:approxA}. Based on Lemma \ref{lemma:optALPHA} and Lemma \ref{lemma:rgtzero}, P3 is solved iteratively to obtain an approximate solution to P1. See Algorithm \ref{non-causal_gen}.} \\ \hline
	\end{tabular}
	\caption{The optimization problems considered in the work  under various cases on the circuit cost, $p$ with some useful comments.}
	\label{my-label}
	\vspace*{-\baselineskip}
\end{table*}

\subsection{On-line Policies} \label{statistical}
In practice, it would be unrealistic to have the non-causal knowledge of the harvested power and the channel state information, but, it is likely that we have statistical information.

The optimization problem P1 in \eqref{eq:average_rate1} does not apply to systems with only stochastic knowledge of the energy arrival profile.  Here, we develop P1 in three directions below, leading to three suboptimal policies to select the decision variables  $\{\boldsymbol{\rho},\boldsymbol{\alpha}_a,\boldsymbol{\alpha}_b,\boldsymbol{\gamma}, \mathbf{d}_b\}$. The motivation behind these simplifications and sub-optimal policies is that they are practical and simple to implement. We also compare the proposed sub-optimal policies with the optimal off-line and on-line policies.
 
\subsubsection{{Optimal Online Policy}}
{To obtain the optimal power allocation when only the causal knowledge and the statistical information of the harvested powers and channel power gains are available, we employ the stochastic dynamic programming based approach \cite{dp}. 
Let $s_n=(C_n,H_n,B_{n-1})$ denote the state of the system in frame $n$, where $C_n$ is the harvested power, $H_n$ is the channel power gain and $B_{n-1}$ is the residual energy at the start of the frame $n$. We assume that the state information of any given frame is known at the start of the frame. Note that $s_1=(C_1,H_1,B_0)$ is the initial state of the system. 
Our goal is to maximize the average rate over a finite horizon of $N$ frames, by choosing a policy, $\pi=\{\rho_n(s_n),\alpha_{a_n}(s_n),\alpha_{b_n}(s_n),d_{b_n}(s_n), \forall s_n, n=1,\ldots,N \}$, that selects time and power splitting ratios and discharge powers for each of the frames. A policy is feasible if the energy causality constraints, battery capacity constraints, bandwidth constraints and  maximum charge and discharge rate constraints, specified in \eqref{eq:gc11} -- \eqref{eq:sc6}, are satisfied for possible states in all the frames. Let $\Pi$ denote the set of all feasible policies. Given the initial state $s_1$, the maximum average rate is given by,
\begin{align}
\mathcal{R}_{\mathrm{on}}^*=\max_{\pi\in\Pi}\mathcal{R}_{\mathrm{on}}(\pi)
\end{align}
where 
\begin{align}
\mathcal{R}_{\mathrm{on}}(\pi)=\frac{1}{N}\sum_{n=1}^{N}\mathbb{E}\left[R(C_n,H_n,B_{n-1},\rho_n,\alpha_{a_n},\alpha_{b_n},d_{b_n})|s_1,\pi\right]
\end{align}
where $R(.)$ is given by \eqref{eq:Rateind} and the expectation is with respect to the random harvested power and the channel power gain. 
The maximum average rate, $\mathcal{R}_{\mathrm{on}}^*$ of the system, given by the value function $J_1(s_1)$, can be computed recursively based on Bellman's equations, starting from $J_N(s_N), J_{N-1}(s_{N-1})$, and so on until $J_1(s_1)$ as follows:
\begin{subequations}\label{eq:dp}
\begin{alignat}{2}
&J_N(C_N,H_N,B_{N-1})=\max_{\{\rho_N,\alpha_{a_N},\alpha_{b_N},d_{b_N}\}} R(C_N,H_N,B_{N-1},\rho_N,\alpha_{a_N},\alpha_{b_N},d_{b_N}) \label{eq:dpN}\\
&J_n(C_n,H_n,B_{n-1})=\max_{\{\rho_n,\alpha_{a_n},\alpha_{b_n},d_{b_n}\}} R(C_n,H_n,B_{n-1},\rho_n,\alpha_{a_n},\alpha_{b_n},d_{b_n})+\bar{J}_{n+1}(C_{n+1},H_{n+1},B_n)\nonumber\\ 
&\qquad\qquad\qquad\qquad\qquad\qquad\qquad\qquad\text{for}\; n=1,\ldots,N-1  \label{eq:dpn}
\end{alignat}
\end{subequations}
where $\bar{J}_{n+1}(C_{n+1},H_{n+1},x)=\mathbb{E}_{C_{n+1},H_{n+1}}\left[{J}_{n+1}(C_{n+1},H_{n+1},x)\right]$ is the average throughput across frames $n+1$ to $N$ averaged over all the realizations of $C_{n+1}$ and $H_{n+1}$. Note that in \eqref{eq:dpn}, we account for the fact that $C_i$'s and $H_i$'s are independent.  Note that the residual energy $B_n$ in \eqref{eq:dpn} is a function of the decision variables $\rho_n,\alpha_{a_n}, \alpha_{b_n}$ and $d_{b_n}$. An optimal policy is denoted as $\pi^*=\{\rho_n^*(s_n),\alpha_{a_n}^*(s_n),\alpha_{b_n}^*(s_n),d_{b_n}^*(s_n), \forall s_n, n=1,\ldots,N \}$, where $\{\rho_n^*(s_n),\alpha_{a_n}^*(s_n),\alpha_{b_n}^*(s_n),d_{b_n}^*(s_n)\}$ is the optimal solution to \eqref{eq:dp} when the state of the system is $s_n$. 
}

\subsubsection{Greedy Algorithm}
When we only have the instantaneous knowledge of the harvested power but not the non-causal or statistical information on the power profile, the entire harvested energy in any frame is utilized in the same frame itself. In each of the frames, the corresponding single frame optimization problem is solved.   Based on Theorem \ref{thm:opt_save_ratio_gen}, the optimal solution can be easily found in each of the frames.   
The algorithm is simple to implement and achieves the optimal rate when all energy in the battery must be used up within each frame. But, the instantaneous optimality comes at the cost of increased circuit energy consumption due to longer duration of circuit operation. 
 \begin{algorithm}[t]
	\caption{ Statistical  Algorithm}
	\label{stat}
	\begin{algorithmic}[1]
		\Procedure{STATISTICAL}{$B_0$, \textbf{c, h,  N}}
		\State Compute $\alpha_{a_i}^*=\argmax_{\alpha_{c_i}\leq \alpha\leq 1} \left( \mathcal{N}_c((1-\alpha)c_i)(1-\alpha)c_i\right)$. $b\leftarrow B_0$
		\For {$i: 1 \rightarrow N$ }
		\State $[\boldsymbol{\rho}^t,\mathbf{d}^t_{b},\boldsymbol{\alpha}_{b}^t]$ =ENERGY-ALLOC($b,[c_i,\bar{C}],[h_i,\bar{H}]$, 2) 
		\State  ${\rho}_i \leftarrow \boldsymbol{\rho}^t(1)$, ${d}_{b_i} \leftarrow \mathbf{d}^t_{b}(1)$ and ${\alpha}_{b_i} \leftarrow \boldsymbol{\alpha}_{b}^t(1)$. 
		\State  $b \leftarrow $ energy remaining in the battery in frame $i$.  
		\EndFor
		\EndProcedure
	\end{algorithmic}
\end{algorithm}
\subsubsection{Statistical Algorithm (SA)}
In addition to the instantaneous knowledge, when we have the statistical information (such as the mean value) of harvested powers and channel gains across the frames, we propose an algorithm based on Algorithm \ref{non-causal_gen}. 
Let the expected values of the harvested power and channel gains be $\bar{C}$ and $\bar{H}$, respectively. Let the TSRs, PSRs and the discharge powers be represented by $(\rho_i,\alpha_{b_i},d_{b_i})$ in frames $i\in \{1,\ldots,N\}$. 

At the beginning of any frame $i$, we have the instantaneous knowledge of the harvested power and the channel gain, i.e., $(c_i,h_i)$, residual energy in the battery and $(\bar{C},\bar{H})$, but, we do not have any  information on $(c_{i+1},\ldots,c_N,h_{i+1},\ldots,h_N)$. To find $(\rho_i,\alpha_{b_i},d_{b_i})$, we consider a hypothetical two-frame optimization problem with the first frame being the frame $i$ and the second frame being a hypothetical frame with parameters $(\bar{C},\bar{H})$.  Then, at the beginning of frame $i$,  for $i=1,\ldots, N$, the transmitter solves the optimization problem  P1 in \eqref{eq:average_rate1} for the above two-frame hypothetical  problem. 
The statistical algorithm is presented in Algorithm \ref{stat}.

\subsubsection{Constant Time/Power Splitting Ratio (CTSR/CPSR) Policies } \label{sec:CSR}
Though the adaptive  policies  described above are simple and practical, simpler systems may not have the capability to measure the harvested energy and the channel states  instantaneously. In such systems, it is not feasible to compute the suitable TSRs and PSRs for each frame instantaneously,  at the frame beginning. It is more practical to use a single, pre-computed time/power splitting ratio across all the frames. A sensible choice of the TSRs and PSRs is the one that maximizes the average rate in \eqref{eq:average_rate1}  with an additional constraint that the TSRs and PSRs in each frame have to be equal i.e., $\rho=\rho_i, \alpha_a=\alpha_{a_i}, \alpha_b=\alpha_{b_i}$, for $i=1,\ldots,N$, with only the stochastic knowledge of energy arrival rates. 

Since, the constraint $(1-\alpha_{b_i})\rho_i=0,\; i=1,\ldots,N$, has to be satisfied, we fix either $\alpha_{b_1}=\ldots=\alpha_{b_N}=1$ and $\alpha_a=\ldots,\alpha_N=\alpha_a^{(c)}=\argmax_{\alpha_{\mathbb{E}(C)}\leq \alpha\leq 1}  \left(\mathcal{N}_c((1-\alpha)\mathbb{E}(C))(1-\alpha)\mathbb{E}(C)\right)$, where $\mathbb{E}(C)$ is the expected value of the harvested power, $\alpha_{\mathbb{E}(C)}=1-C_p/\mathbb{E}(C)$ and, find the optimal $\rho$, leading to \emph{constant time splitting ratio (CTSR) policy }, or fix $\rho_1=\ldots=\rho_N=0$ and find the optimal $\alpha_b$, leading to \emph{constant power splitting ratio (CPSR) policy }.

\section{Numerical Results} \label{numerical_results}
In our simulations, we assume that information rate $R=0.5\log(1+hP_t/(N_0W))$ bits per channel use, where $P_t$ is the average transmit power, $h$ is the channel power gain, $N_0=10^{-15}$ \si{\watt/\hertz} is the AWGN power spectral density and $W=\SI{1}{\mega \Hz}$ is the channel bandwidth. We assume that $N_s=10^6$ symbols are transmitted per frame of duration $T=\SI{1}{\second}$. 
We assume that the channel power gain is exponentially distributed with unit mean. 
Based on \cite{Krieger}, we assume that the charging efficiency,  $\mathcal{N}_c(c_p)=1.5-0.5\sqrt{1+4rc_p/V_B^2}$, the discharge efficiency, $\mathcal{N}_d(d_p)=0.5+0.5\sqrt{1-4rd_p/V_B^2}$, the maximum charge power, $C_p=2V_B^2/r$ and the maximum discharge power, $D_p=V_B^2/(4r)$, where 
$V_B$ is the nominal voltage of the battery. 
   \begin{figure*}[t]
	\centering
	\begin{subfigure}{0.48\textwidth}
		\includegraphics[width=1\textwidth]{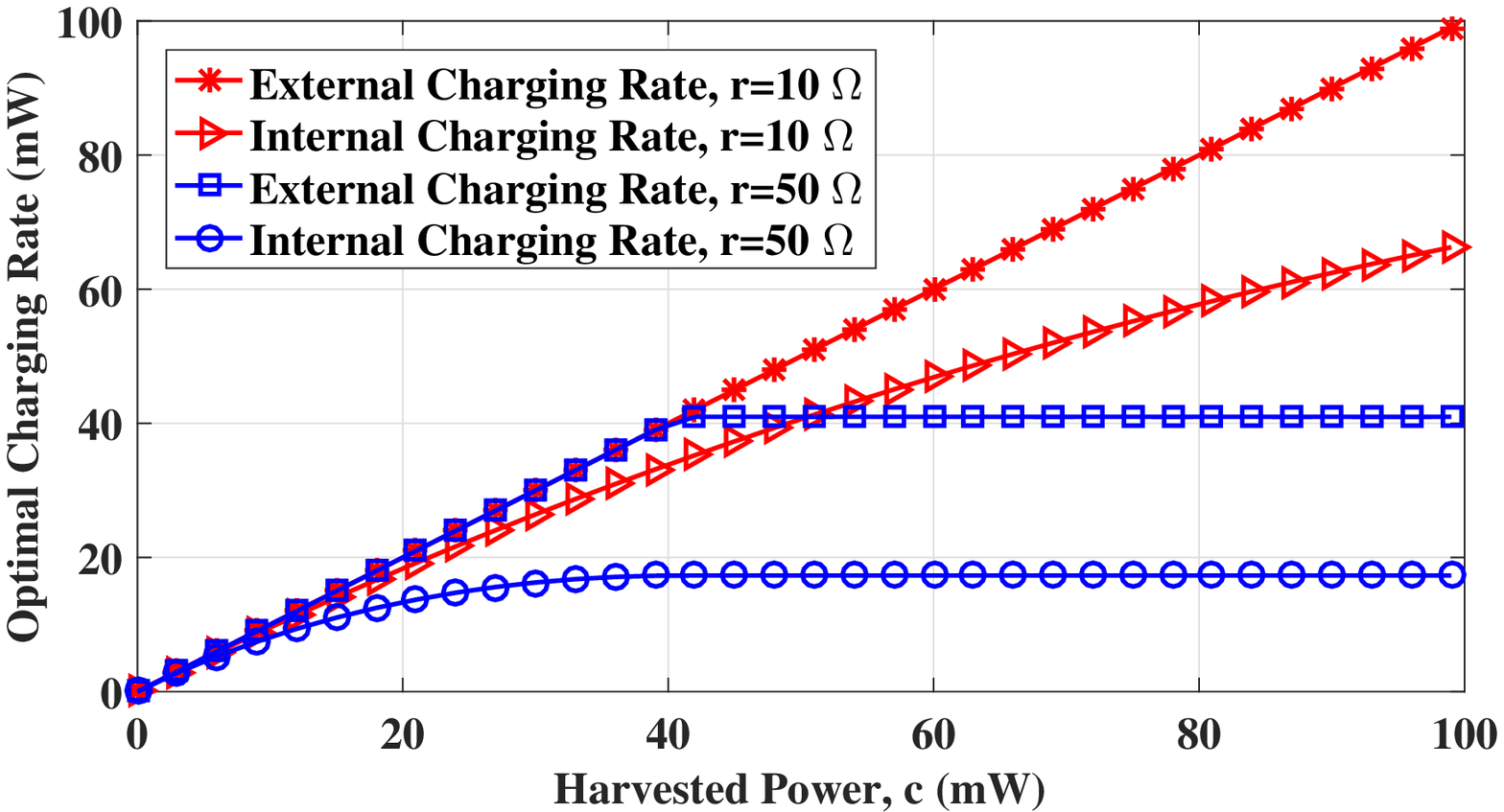}
		\caption{Variation with the harvested power for $h=1$, $V_B=\SI{1.5}{\volt}$.}
		\label{fig:charge_rateH}
	\end{subfigure}
	\begin{subfigure}{0.48\textwidth}
		\includegraphics[width=1\textwidth]{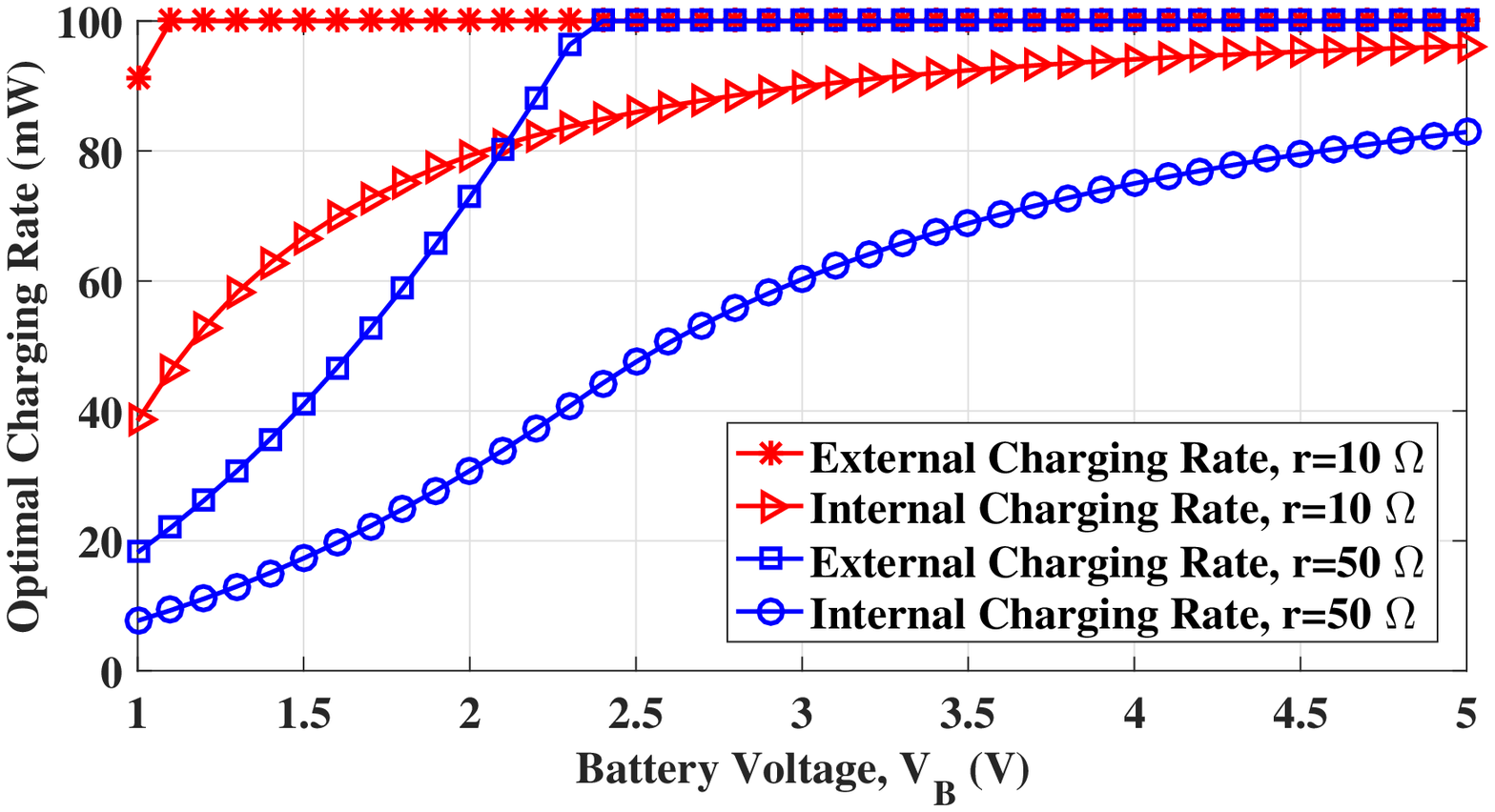}
		\caption{Variation with the battery voltage for $h=1$, $c=\SI{100}{\milli\watt}$.}
		\label{fig:charge_rateV}
	\end{subfigure}%
	\caption{Variation of optimal internal and external charging rates  with  (a) the deterministic harvested power and  (b) the battery voltage.}   
	\label{fig:charge_rate}
	\vspace{-.5cm}
\end{figure*}
\subsection{Variation of  Optimal Charging Rates with the Harvested Power and Nominal Battery Voltage}
In Fig. \ref{fig:charge_rate}, we present the variation of optimal internal and external charging rates with harvested power (in Fig. \ref{fig:charge_rateH}) and the nominal battery voltage (in Fig. \ref{fig:charge_rateV}). The external charging rate, given by $c_p^*=(1-\alpha_a^*)c$, indicates the power directed to the battery after the optimal power splitting, and the internal charging rate,  given by $c_p^*\mathcal{N}_c(c_p^*, r)$, indicates the the rate at which energy gets stored in the battery internally, after the losses in the internal resistance. 

We make two important observations from Fig. \ref{fig:charge_rateH}. First, when the internal resistance is \emph{low},  the external charging rate linearly increases with the harvested power (in this case, $\alpha_a^*=0$), but the internal charging rate increases at a slower rate with the harvested power due to the resistive losses.  Second, when the internal resistance is \emph{high},  then both the external and internal charging rates increase only up to a threshold, beyond which the battery is charged at the optimal charging rate, which is independent of the harvested power. Similarly, from Fig. \ref{fig:charge_rateV}, we note that the internal resistance significantly  impacts the external and internal  charging rates for a wide range of the nominal voltage of the battery.   
   \begin{figure*}[t]
	\centering
	\begin{subfigure}{0.48\textwidth}
		\includegraphics[width=\textwidth]{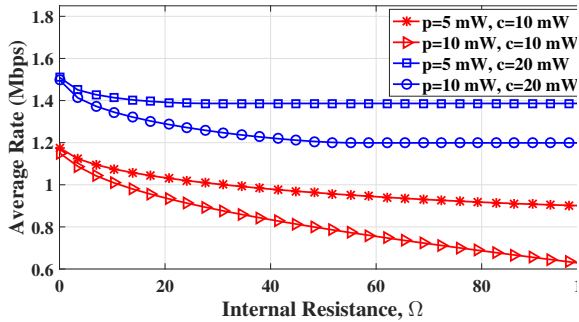}
		\caption{Optimal Rate for the Single Frame}
		\label{fig:inst_rate}
	\end{subfigure}
	\begin{subfigure}{0.48\textwidth}
		\includegraphics[width=\textwidth]{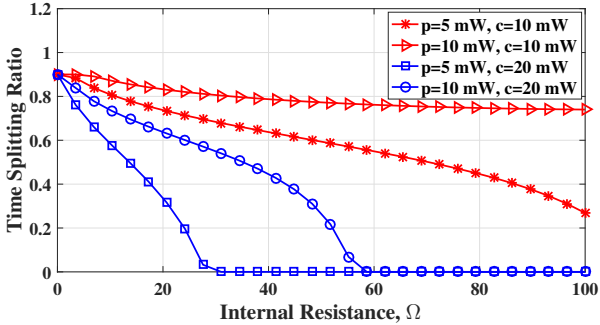}	
		\caption{Optimal Time Splitting Ratio for the Single Frame}
		\label{fig:inst_save_ratio}
	\end{subfigure}%
	\caption{Comparison of (a) optimal rates and the corresponding (b) optimal time splitting ratios and their variation   with the internal resistance for two values of  the circuit power, $p$ and non-random harvested powers, $c$ with  $T=\SI{1}{\second}$, $h=1$, $B=\SI{20}{\milli\joule}$,  $V_B=\SI{1.5}{\volt}$ and $\rho_{W}=0.9$.}   
	\label{fig:inst}
	\vspace{-.7cm}
\end{figure*}

\subsection{The Optimal Rate in the Single Frame Case }  
Fig. \ref{fig:inst_rate} shows the variation of the optimal rate with the battery internal resistance for two values of the circuit cost and harvested powers and Fig. \ref{fig:inst_save_ratio} shows the corresponding optimal TSRs.   
From Fig. \ref{fig:inst}, we make two important observations. 
First, when the circuit cost is in the order of the harvested power,  the optimal rate decreases with the increasing internal resistance and the optimal TSR is greater than zero. This is because when the circuit cost is in the order of the harvested power, one can save on the circuit losses by operating the circuit for smaller amount of time while the harvested energy is stored and drawn from the battery. 
Second, when the harvested power is few times more than the circuit cost, then the optimal rate decreases up to a certain point beyond which the rate is independent of the internal resistance.  The reason is that the battery charge and discharge losses increase as the internal resistance increases. But, the system continues the transactions (charging and discharging) with the  battery to reduce the circuit losses up to a certain point. This can be seen from Fig. \ref{fig:inst_save_ratio} where the TSR is greater than zero up to a certain value of the internal resistance.  As the internal resistance increases, the battery charging and discharging losses surpass the gain obtained by avoiding the circuit losses and, it turns out that avoiding any transactions with the battery is optimal.  Obviously, the optimal rate after the \emph{cut-off} point is independent of battery parameters.   

\subsection{Optimal Transmit Power Levels Under Two Different Models for Battery Losses }

   \begin{figure}[t]
	\centering
	\includegraphics [width=.75\textwidth] {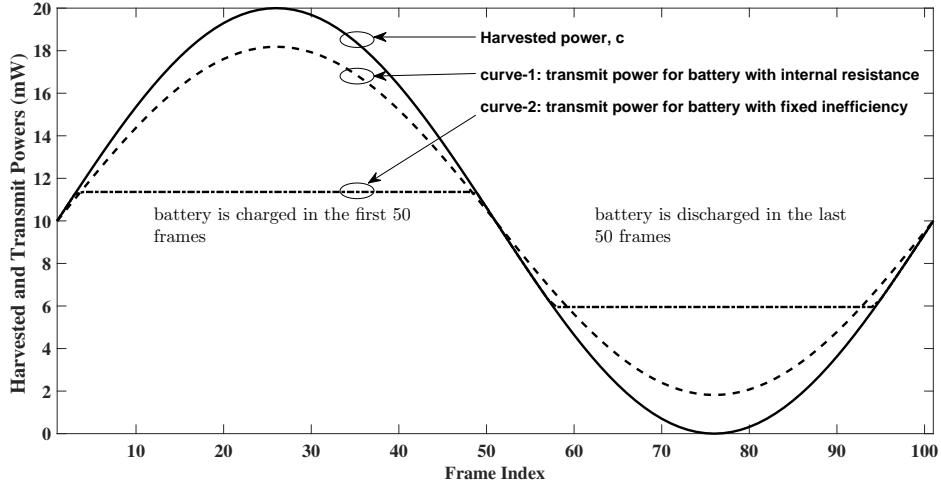}	
	\caption{A comparison of the optimal transmit power levels in two different battery loss models.   For curve-1, we assume that the battery has a non-zero internal resistance (\SI{5}{\ohm}) which results in battery charge/discharge inefficiencies that are functions of charge/discharge rates. The curve-2 is plotted assuming that the battery  has a fixed round-trip efficiency ($\mathcal{N}=\mathcal{N}_c\mathcal{N}_d=0.75$) as in  \cite{inefficiency}.  We assume that $B=\infty$ and  $V_B=\SI{1.5}{\volt}$.  }   
	\label{fig:ineff}
	\vspace{-.7cm}
\end{figure}
Due to the non-zero internal resistance, charge/discharge efficiencies vary with charge/discharge rates. 
Hence, it is insightful to compare optimal power allocation in this case with that when the  battery efficiency is a constant as in \cite{inefficiency}.  We show the optimal transmit powers in these two cases with the harvested power in Fig. \ref{fig:ineff}.  
It has been shown in \cite{inefficiency} that the optimal power allocation has a double threshold structure as shown in curve-1 of  Fig. \ref{fig:ineff}.  Unlike in curve-1, it is interesting to note that if the battery has a non-zero internal resistance,  the optimal transmit power strictly monotonically increases with the harvested power as shown in curve-2 and proved in Theorem  \ref{thm:zeropower}.

\subsection{Variation of the Average Rate in the Off-Line Policy with the Average Harvested Power}\label{PsvsTs}

\begin{figure}[t]
	\centering
	\includegraphics [width=.7\textwidth] {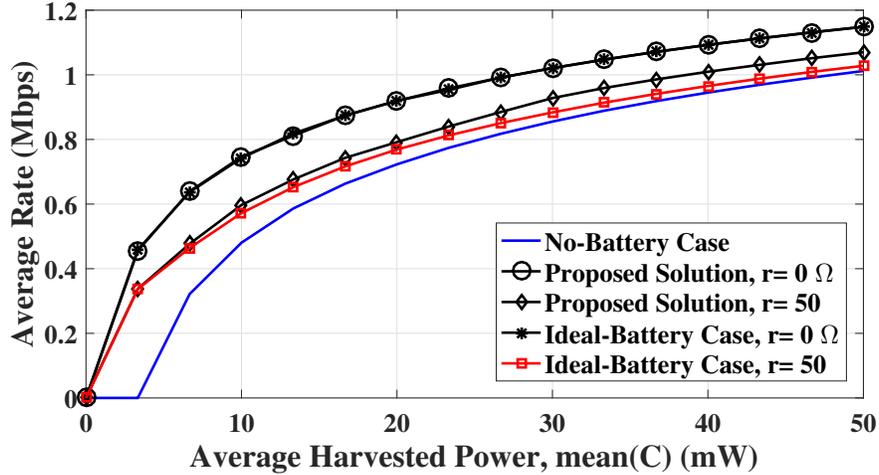}	
	\caption{Variation of the average rate with the average harvested power for various cases in the off-line policy with $T=\SI{1}{\second}$, $B=\SI{100}{\milli\joule}$,  $V_B=\SI{1.5}{\volt}$, $p=\SI{10}{\milli\watt}$, $N=100$ and $\rho_{W}=0.9$.}   
	\label{fig:offline}
		\vspace{-.7cm}
\end{figure}

In Fig. \ref{fig:offline}, we present the variation of the average rate in the off-line policy with the average harvested power, obtained by averaging the numerical results from 1000 independent runs of Monte Carlo simulations. The No-Battery Case curve assumes that the system is not equipped with any battery and the Ideal-Battery Case curve is obtained by adopting the optimal policy in an ideal battery (with zero internal resistance) to the non-ideal battery case. 
Note that the rate of increase in the average rate with the average harvested power is considerably affected by the  internal resistance.  This is because as the average harvested power increases, the charging/discharging rates increase resulting in the increased charging/discharging losses. It is interesting to note that as the average harvested power increases,  the average rate in the Ideal-Battery Case approaches the average rate in the No-Battery Case implying that the optimal policies designed for an ideal battery may be strictly suboptimal when the internal resistance is non-zero. 
\begin{figure}[b]
	\centering
	\includegraphics [width=.7\textwidth] {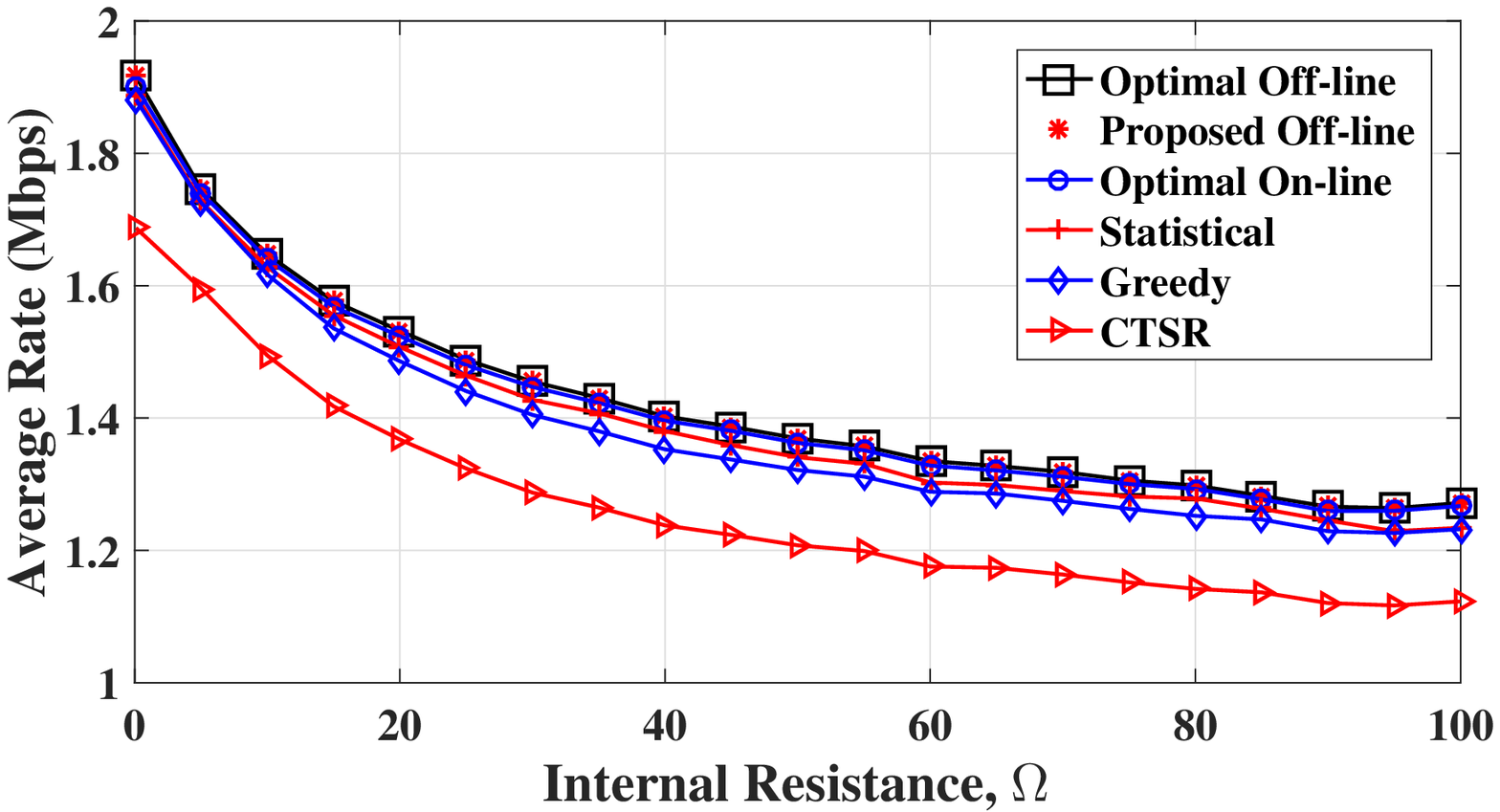}	
	\caption{Variation of the average rates with the internal resistance in the various algorithms with $p=\SI{50}{\milli\watt}$,  $T=\SI{1}{\second}$, $B=\SI{100}{\milli\joule}$,  $V_B=\SI{1.5}{\volt}$, $N=5$ and $\rho_{W}=0.9$.}   
	\label{fig:online}
	\vspace{-.7cm}
\end{figure}
\subsection{Comparison of the Performances of On-line and Off-line Policies}\label{OnvsOff}
{In Fig. \ref{fig:online}, we plot the average rates in the  off-line and on-line policies against the internal resistance values. We assume that the harvested power is uniformly distributed in $\{\SI{50}{\milli\watt}, \SI{100}{\milli\watt}\}$. In the dynamic programming based  Optimal On-line policy, $B_n$'s are discretized in step sizes of $0.0005$. We obtain the average rates by averaging the numerical results from $10^4$ independent runs of Monte Carlo simulations. To reduce the computational complexity, we use the step discharge model for all the algorithms. We fix the number of frames, $N$ to $5$ as the computational complexity becomes prohibitive for a larger N. }

{There are several interesting points to note in  Fig. \ref{fig:online}. 
First, average rates in all the policies decrease with the increasing internal resistance thereby indicating that the internal resistance is an important battery parameter that affects the performance of the EH-based communication systems significantly. Further, we note that the performance of the Proposed Off-line policy and the original Off-line policy in P1 are almost the same. As expected, the performance of the Optimal On-line policy is slightly worse than that of the off-line policies. Further, we note that the performance of the Statistical algorithm which is based on the Proposed Off-line algorithm is close to that of the Optimal On-line and  off-line policies. The CTSR algorithm performs worse than the Greedy algorithm as the CTSR algorithm does not adapt its decision variables to the varying harvested power. The average rate in CPSR algorithm is equal to $1.02$ Mbps and it is independent of the internal resistance, i.e., $\alpha_{b}^*=1$ which implies that the harvested energy is not stored in the battery.   }
{\begin{table}[t]
		\centering
		\begin{tabular}{|c|c|c|c|c|}
			\hline	
			Algorithms	      & $N=25$                & $N=50$              & $N=75$                & $N=100$                 \\ \hline
			Proposed Offline Algorithm    & \SI{0.3}{\second}     & \SI{1.27}{\second}  &  \SI{2.85}{\second}   & \SI{7.79}{\second}      \\ \hline
			Statistical Algorithm & \SI{0.26}{\second}    & \SI{.14}{\second}   & \SI{0.1}{\second}     & \SI{0.08}{\second}      \\ \hline
			Greedy Algorithm      & \SI{10}{\milli\second}& \SI{9}{\milli\second}     & \SI{5}{\milli\second} & \SI{4}{\milli\second}   \\ \hline
			CTSR Algorithm        & \SI{13}{\milli\second}& \SI{10}{\milli\second}    & \SI{10}{\milli\second}& \SI{10}{\milli\second}   \\ \hline
			CPSR Algorithm	      & \SI{10}{\milli\second}& \SI{10}{\milli\second}    & \SI{15}{\milli\second}& \SI{15}{\milli\second}   \\ \hline
		\end{tabular}
		\caption{Normalized runtime of the algorithms with Intel i7-5600U processor running at \SI{2.60}{\giga\hertz} using \textsc{Matlab} 2016a software package.  }
		\label{tab:running time}
		\vspace*{-\baselineskip}
\end{table}}
\subsection{Variation of the Normalized Runtime of the Algorithms with the Number of Frames}\label{runtime}
{
	In Table \ref{tab:running time}, we present the normalized runtime, the ratio of the total time taken to run the algorithm to the number of frames, for various algorithms when the algorithms are executed in Intel i7-5600U processor running at \SI{2.60}{\giga\hertz} using \textsc{Matlab} 2016a software package.  For the Proposed Off-line algorithm, as expected from the  complexity analysis, the runtime scales with the number of frames as $N^3$, approximately. Further, as expected, in all other algorithms, the normalized runtime does not change significantly with $N$. In the Statistical algorithm, one may note that the normalized runtime slowly decreases as $N$ increases. This is because the overhead of the algorithm dominates the runtime when the number of frames is small.}

\section{Conclusions}\label{conclusion}
In this paper,  we argue that the battery internal resistance fundamentally changes the way we design energy management techniques in energy harvesting communication systems. Our study shows that the internal resistance considerably inhibits the energy redistribution across frames.  This causes a significant reduction in the optimal average communication rate compared to that obtained using an ideal battery (i.e., zero internal resistance).  Furthermore, the optimal policy designed for an ideal battery performs poorly when the internal resistance is not negligible. 

In our work, the charging/discharging efficiencies are modeled as functions of the internal resistance and charge/discharge powers.  We assume a finite capacity battery,  non-zero circuit power and take into account limitations on bandwidth. In this context, we derive compact expressions for optimal time and power splitting ratios in the single frame case.
We then propose an iterative off-line algorithm to approximately solve the non-convex optimization problem which assumes \emph{a priori} knowledge of the harvested powers and channel gains in the multiple frame case. {We also solve for the optimal on-line policy by using stochastic dynamic programming assuming statistical knowledge and causal information of the harvested power and channel power gain variations.} We then propose three heuristic on-line algorithms and show that an algorithm that is inspired by the off-line policy performs significantly better than the other two heuristic algorithms. {Advanced analysis of the proposed algorithms is considered as a future work. }

\section*{Appendix}
\subsection{Proof of Lemma \ref{lemma:transmitpower} }
When $\alpha_a c \leq p$, we cannot operate the circuit during $[0,\rho\tau)$ for any $\rho$, hence, $\max(\alpha_ac-p,0)=0$ and  $R_a=0$ for any $\rho$, including $\rho=\rho^*$ and the transmission occurs only over $[\rho\tau,\tau]$ with constant power $\tau(\alpha_b c-p+d_b)(1-\rho)/N_s$.  
When $\alpha_a c >p$ we have,
\begin{align}
R(\rho,\alpha_a,\alpha_b,\gamma,d_b)&=\gamma R_a +(1-\gamma) R_b 
=\gamma \log(1+hP_a)+(1-\gamma)\log(1+hP_b)\\
&\stackrel{\text{a}}{\leq} \log(1+h\tau/N_s\left(\rho(\alpha_a c-p) +(1-\rho)(\alpha_b c-p+d_b)\right)) \\&\stackrel{\text{b}}{\leq}  \log(1+h\tau/N_s\left(c-p+ \hat{d}_b\right))
= \left.R_b\right\vert_{\rho=0} 
\end{align}
where $\hat{d}_b=\{d:d\tau/\mathcal{N}_d(d)=B_0\}$,  (a) follows from Jensen's inequality, (b) holds because of the following. When $d_b=0$, the term, $\rho(\alpha_a c-p) +(1-\rho)(\alpha_b c-p)\leq c(\rho\alpha_a+(1-\rho)\alpha_b)-p\leq c-p $ as $\rho\alpha_a+(1-\rho)\alpha_b\leq 1$; for any $d_b>0$, we have $\alpha_b=1$ and, from \eqref{eq:sc1}, $d_b(1-\rho)\tau =\mathcal{N}_d(d_b)(B_{\rho\tau}+B_0)=\mathcal{N}_d(d_b)\mathcal{N}_c((1-\alpha_a)c)(1-\alpha_a)c\rho\tau+\mathcal{N}_d(d_b)B_0 \leq(1-\alpha_a)c\rho\tau+\mathcal{N}_d(\hat{d}_b)B_0$  which implies that $\rho(\alpha_a c-p) +(1-\rho)(\alpha_b c-p+d_b)\leq c-p+\mathcal{N}_d(\hat{d}_b)B_0/\tau=c-p+\hat{d}_b$, where the upper bound is attained when $\rho=0$. 

Consolidating the results, for any given $\rho$ and a policy that has $R_a>0$, we can always find another policy with a higher average rate such that  $R_a=0$,  while satisfying all the constraints. Hence, $\gamma^*=0$ and $R_{[0,\rho^*\tau)}=0$ and all $N_s$ symbols are transmitted only during $[\rho^*\tau,\tau]$ with the transmit power $\tau(\alpha_b^* c-p+d_b^*)(1-\rho^*)/N_s$.      Lemma \ref{lemma:transmitpower} is thus proved.
\vspace{-.25cm}
\subsection{Proof of Theorem \ref{thm:opt_save_ratio_gen} }
Based on Lemma \ref{lemma:transmitpower}, we have $R=\log(1+h\left(\alpha_b c+d_b-p\right)(1-\rho)\tau/N_s)$. Since, $R$ is a monotonically increasing function of the transmit power, in order to maximize $R$, we can simply maximize the transmit power, $P=\left(\alpha_b c+d_b-p\right)(1-\rho)\tau/N_s$.  Since, $\tau$ and $N_s$ are constants, we instead maximize $E(\alpha_a, \alpha_b, d_b,\rho)=(\alpha_b c+d_b-p)(1-\rho)$. 
We first solve the problem by relaxing the  battery capacity constraint in \eqref{eq:sc2}. 
Since, draining the battery completely is optimal and, noting that $(1-\alpha_b)d_b=0$, in the optimal policy we must have, $(d_b/ \mathcal{N}_{d}(d_b))=(\mathcal{N}_c((1-\alpha_a )c)(1-\alpha_a )c\rho\tau+B_0)/((1-\rho)\tau)$ and $d_b\leq D_p$, for any feasible $\alpha_a$ and $\rho$, from \eqref{eq:sc1} and \eqref{eq:sc4}, respectively. 
Define $\hat{d}_b({\alpha_a,\rho})=\{\min(d_b,D_p):d_b/\mathcal{N}_{d}(d_b)= (\mathcal{N}_c((1-\alpha_a )c)(1-\alpha_a )c\rho\tau+B_0)/((1-\rho)\tau)\}$. For a concave decreasing $\mathcal{N}_d(d_b)$, it can be shown that $d_b/\mathcal{N}_d(d_b)$ is a convex increasing function of $d_b$. Hence, if $\alpha_a$ and $\rho$ are given, we can uniquely determine $\hat{d}_b(\alpha_a,\rho)$ always.  Hence, the optimal discharge power $d_b^*=\hat{d}_b({\alpha_a^*,\rho^*})$. Now, $E(.)$ can be treated as a function of only $\alpha_a, \alpha_b$ and $\rho$.  
Hence, $E(\alpha_a, \alpha_b,\rho)=(\alpha_b c-p+\hat{d}_b({\alpha_a,\rho}))(1-\rho)$.  
For any $\rho$ and $\alpha_a$, the quantity  $\alpha_bc(1-\rho)$ achieves its maximum at $\alpha_b=1$, hence, $\alpha_b^*=1$. Further, $\hat{d}_b({\alpha_a,\rho})$ is a monotonic increasing function of  $\mathcal{N}_c((1-\alpha_a )c)(1-\alpha_a )c$ and hence, it attains the maximum at $\alpha_a =\alpha_a^*=\argmax_{\alpha_c\leq \alpha\leq 1} (\mathcal{N}_c((1-\alpha )c)(1-\alpha )c)$ for any $\rho$. 
Hence, $E(\alpha_a^*, \alpha_b^*,\rho)=(\alpha_b c-p+\hat{d}_b(\alpha_a^*,\rho))(1-\rho)$. To obtain the maximum rate, we simply need to maximize $E(\alpha_a^*, \alpha_b^*,\rho)$ over $\rho$. 
Now, we note that the battery capacity constraint simply puts an upper bound on $\rho$.  From \eqref{eq:sc2}, we have, $\mathcal{N}_c(c_p^*)c_p^*\rho\tau+B_0 \leq B$ which implies $ \rho \leq (B-B_0)/(\mathcal{N}_c(c_p^*)c_p^*\tau) $. 
Hence the proof.   

\vspace{-.25cm}
\subsection{Proof of Lemma \ref{lemma:gamma}}
Recall that in Section \ref{single frame} we had defined $\gamma_i=0$ if $(\alpha_{a_i}c_i-p)\rho_i\tau<0$. 
Hence, to prove $\gamma_i^*=0$ for any $i$ we need to prove  $(\alpha_{a_i}^*c_i-p)^+\rho_i^*\tau=0$ for the $i$-th frame, where $(x)^+=\max(x,0)$. 
 If $(\alpha_{a_i}^*c_i-p)\leq 0$, then, we always have $(\alpha_{a_i}^*c_i-p)^+\rho_i^*\tau=0$.  But, whenever $(\alpha_{a_i}^*c_i-p)>0$,  we need to prove that $\rho_i^*=0$. To accomplish this, we note that the decision variables are coupled across the various frames as energy may get transferred from one frame to another in the optimal policy. This energy transfer can be accounted for by considering the residual energy available in the battery at the start of each of the frames. Let $B_{i-1}$ be the stored energy in the battery at the start of any frame $i$. Then, the energy consumed by frame $i$ from the battery is $B_{i-1}-B_{i}$ (a negative value indicates that energy is stored in the battery) in any frame  $i$.  By some means, if we know the value of  $B_{i-1}$'s, then, we can optimize each frame independent of the other frames. Assume that for any  frame $i$, $\alpha_{a_i}^*$, $\rho_i^*>0$, $B_{i-1}^*$ and $B_{i}^*$ are the optimal values.  As in the proof of Lemma \ref{lemma:transmitpower}, we can show that for any  frame $i$ with $\alpha_{a_i}^*c>p$ and $\rho_i^*>0$,  for any $B_{i-1}^*$ and $B_{i}^*$ values,  we can achieve a higher rate in frame $i$, than the rate when $\rho_i^*>0$, by selecting $\rho_i'^*=0$ and choosing an arbitrary $\alpha_{a_i}'^*$.  Hence, by contradiction, we must have $\rho^*_i=0$ in the optimal policy.  This proves that $\gamma_i^*=0$ for any $i$ in the optimal policy.

\subsection{Proof of Theorem \ref{thm:zeropower} }
We first note that whenever the battery capacity is infinite, \eqref{eq:gc12} is inactive and P2 is convex. Hence,  Karush-Kuhn-Tucker (KKT) conditions are necessary and sufficient for optimality.  
The Lagrangian of P2 is given by
\begin{align}
&L2=-\frac{1}{N}\sum_{i=1}^{N}\log\left(1+h_i(\alpha_{b_i}c_i+d_{b_i})\tau/N_s\right)+\sum_{i=1}^{N}\lambda_i\left(\sum_{k=1}^{i}\left(\tilde{d}_{b_k}-\tilde{c}_{b_k}\right)\tau-B_0\right) \nonumber\\
& -\sum_{i=1}^{N}\omega_id_{b_i}+\sum_{i=1}^{N}\delta_i(d_{b_i}-D_p)-\sum_{i=1}^{N}\mu_i(\alpha_{b_i}-\alpha_{c_i})+\sum_{i=1}^{N}\nu_i(\alpha_{b_i}-1)
\end{align}
where $\lambda_i, \omega_i, \delta_i, \mu_i$ and $\nu_i$ are non-negative Lagrange multipliers corresponding to inequalities \eqref{eq:gc11}, $d_{b_i}\leq 0$, $d_{b_i}\leq D_p$, $\alpha_{c_i}-\alpha_{b_i}\leq 0 $  and $\alpha_{b_i}-1\leq 0$, respectively. 
We first consider the case when the battery is being charged. 
The stationary conditions imply that
\begin{align}\label{eq:txpower}
P_i=\frac{\left(\alpha_{b_i}c_i+d_{b_i}\right)\tau}{N_s}&=\frac{c_i\tau/\ln(2)}{-\tilde{c}_2'(\alpha_{b_i})\tau NN_s\sum_{j=i}^{N}\left(\lambda_j\right)-\mu_i+\nu_i}-\frac{1}{h_i}
\end{align}
Now, consider any two frames $j$ and $k(>j)$ such that the battery has a non-zero amount of residual energy less than its capacity in all the frames between them. 

For any frame $i$, since the battery is charged at $(1-\alpha_{b_i})c_i$ W,  we must have $\alpha_{c_i}\leq \alpha_{b_i}<1$.  The transmit power $P_i=\alpha_{b_i}c_i\tau/N_s$ as $d_{b_i}=0$. Whenever $P_i>0$ and when the battery is charged at the rate strictly less than $C_p$, we must have $\alpha_{b_i}>\alpha_{c_i}$.  Hence, from complementary slackness conditions, we have $\mu_i=\nu_i=0$. From \eqref{eq:txpower}, after rearranging the terms, we have,
\begin{align}\label{eq:ga}
g_{\alpha}(\alpha_{b_i},c_i)=-\tilde{c}_2'(\alpha_{b_i})\alpha_{b_i}\tau N+\frac{-\tilde{c}_2'(\alpha_{b_i})NN_s}{h_ic_i}=\frac{1}{\ln(2)\sum_{j=i}^{N}\left(\lambda_j\right)}
\end{align} 
Note that in the frame $i$ between any two frames in which the battery is fully drained, $\lambda_i=0$. Hence, the right hand side in \eqref{eq:ga} and consequently, $g_{\alpha}(\alpha_{b_i},c_i)$ remain constant. 
Recall that $\tilde{c}_2(\alpha_{b_i})=(1-\alpha_{b_i})c_i\mathcal{N}_c((1-\alpha_{b_i})c_i)$ which implies $-\tilde{c}_2'(\alpha_{b_i})=\mathcal{N}_c((1-\alpha_{b_i})c_i)c_i-(1-\alpha_{b_i})c_i\mathcal{N}_c'((1-\alpha_{b_i})c_i)$.  Hence, for any $c_k>c_j$, from \eqref{eq:ga}, we have, 
\begin{align}
&(\mathcal{N}_c((1-\alpha_{b_j})c_j)-(1-\alpha_{b_j})\mathcal{N}_c'((1-\alpha_{b_j})c_j))\left(\alpha_{b_j}c_j\tau N+ A\right)=\nonumber\\
&(\mathcal{N}_c((1-\alpha_{b_k})c_k)-(1-\alpha_{b_k})\mathcal{N}_c'((1-\alpha_{b_k})c_k))\left(\alpha_{b_k}c_k\tau N+ A \right)\label{eq:eq}\\
&(\mathcal{N}_c((1-\alpha_{b_j})c_j)-(1-\alpha_{b_j})\mathcal{N}_c'((1-\alpha_{b_j})c_j))\left(\alpha_{b_j}\tau N+A \right)>\nonumber\\
&(\mathcal{N}_c((1-\alpha_{b_k})c_k)-(1-\alpha_{b_k})\mathcal{N}_c'((1-\alpha_{b_k})c_k))\left(\alpha_{b_k}\tau N+A\right) \label{eq:in}
\end{align} 
where $A={NN_s}/{h}$.
Now, by contradiction, we can prove that $\mathcal{N}_c((1-\alpha_{b_j})c_j) >  \mathcal{N}_c((1-\alpha_{b_k})c_k)$ (if  $\mathcal{N}_c((1-\alpha_{b_j})c_j) \leq \mathcal{N}_c((1-\alpha_{b_k})c_k)$, it contradicts \eqref{eq:in}). Substituting this result in \eqref{eq:txpower}, it can be shown that that $P_k>P_j$ for any $c_k>c_j$. Further, when the battery is charged at its maximum charge rate of $C_p$, the result follows straightforward as the excess power is directly used for the transmission from the direct path. 

Using the similar technique, we can derive the result when the battery is discharged at a rate below the maximum discharge rate $D_p$ in the optimal case.  If the battery discharge rate is fixed at $D_p$ in the optimal case, the result is straightforward as the harvested power is directly used for the transmission from the direct path. Hence, the proof.   
\vspace{-.25cm}
\subsection{Proof of Lemma \ref{lemma:optALPHA}}\label{app:lemma:optALPHA}
For a given frame $i$ with a given $\rho_i$, from Theorem  \ref{thm:opt_save_ratio_gen}, it follows that  $R_i(\alpha_{a_i})\leq R_i(\alpha_{a_i}^*)$ if $\rho_i>0$ without impacting the rates in the other frames.  If $\rho_i=0$,  we note that the value of $\alpha_{a_i}$ does not play any role in the optimization problem.  
The above two statements hold true irrespective of the optimal values in the other frames. Hence, the result follows. 

\subsection{Proof of Lemma \ref{lemma:rgtzero} }\label{app:opt_save_ratio_gen}
For simplicity, we assume that  the battery capacity constraint in \eqref{eq:gc12} is inactive. 
Let $\mathcal{N}_c^{(y)}=(1-y)\mathcal{N}_c((1-y)c)$. 
From Lemma \ref{lemma:optALPHA}  we have $\alpha_{a_i}^*=\argmax_{\alpha_{c_i}\leq \alpha\leq 1} (\mathcal{N}_c((1-\alpha)c)(1-\alpha)c)$.
Let $\rho_i>0$ and $\alpha_{b_i}<1$ be the optimal solution for any frame $i$. 
When $d_b>0$, based on our remarks in the system model, we must have $\alpha_{b_i}=1$.   Hence, we cannot have $\alpha_{b_i}<1$ in the optimal solution. 
When $d_b=0$,  let $B_{i-1}$ and $B_{i}$ be the residual energy at the start of frame $i$ and $i+1$, respectively. Hence,
\begin{align}\label{eq:chargeab}
\mathcal{N}_c((1-\alpha_{a_i})c)(1-\alpha_{a_i})c\rho_i\tau+\mathcal{N}_c((1-\alpha_{b_i})c)(1-\alpha_{b_i})c(1-\rho_i)\tau=B_{i}-B_{i-1}
\end{align} 
Now, let us consider $\alpha_{b_i}'=1$ with the corresponding  $\rho_i'>\rho_i$ such that
\begin{align}\label{eq:chargea}
\mathcal{N}_c((1-\alpha_{a_i})c)(1-\alpha_{a_i})c\rho'\tau=B_{i}-B_{i-1}
\end{align}

From \eqref{eq:chargeab} and \eqref{eq:chargea}, we have $\rho_i'=\rho_i+(1-\rho_i)\frac{\mathcal{N}_c^{(\alpha_{b_i})}}{\mathcal{N}_c^{(\alpha_{a_i})}}$.
  Let $E_{\alpha_{b_i}'=1,\rho_i'}$ and $E_{\alpha_{b_i},\rho_i}$ be the transmit energy  $\alpha_{b_i}'$ and $\alpha_{b_i}$, respectively.  Now, consider the difference of transmit energy in the two cases, i.e.,  
\begin{align}
&\frac{N_s}{\tau}(E_{\alpha_{b_i}'=1,\rho_i'}-E_{\alpha_{b_i},\rho_i})=(c_i-p)(1-\rho_i')-(\alpha_{b_i} c_i-p)(1-\rho_i)\\
&\stackrel{\text{a}}{=} (c_i-p)(1-\rho_i-(1-\rho_i)\frac{\mathcal{N}_c^{(\alpha_{b_i})}}{\mathcal{N}_c^{(\alpha_{a_i})}})-(\alpha_{b_i} c_i-p)(1-\rho_i)\\
&=(1-\rho_i)\left(   c_i\left(1-\alpha_{b_i}-\frac{\mathcal{N}_c^{(\alpha_{b_i})}}{\mathcal{N}_c^{(\alpha_{a_i})}}\right)+ p\left(\frac{\mathcal{N}_c^{(\alpha_{b_i})}}{\mathcal{N}_c^{(\alpha_{a_i})}}\right)\right)
\end{align}
\begin{align}
&=(1-\rho_i)\frac{c_i\mathcal{N}_c^{(\alpha_{b_i})}}{\mathcal{N}_c^{(\alpha_{a_i})}}  \left(\frac{\mathcal{N}_c((1-\alpha_{a_i})c)(1-\alpha_{a_i})}{\mathcal{N}_c((1-\alpha_{b_i})c_i)}-1+\frac{p}{c_i} \right)\\
&\stackrel{\text{b}}{\geq}(1-\rho_i)\frac{c_i\mathcal{N}_c^{(\alpha_{b_i})}}{\mathcal{N}_c^{(\alpha_{a_i})}} \left(1-\frac{p}{c_i}\right)\left(\frac{1-\mathcal{N}_c((1-\alpha_{b_i})c_i)}{\mathcal{N}_c((1-\alpha_{b_i})c_i)}\right) \stackrel{\text{c}}{\geq} 0
\end{align}
where (a) is obtained by substitution of $\rho_i'$, (b) follows from Theorem  \ref{thm:opt_save_ratio_gen}, and (c) follows because   $1-\mathcal{N}_c((1-\alpha_{b_i})c_i) \geq 0$ and because of the fact that we cannot run the circuitry if $c+d_{b_i}\leq p$. Hence, the transmit energy when $\alpha_{b_i}=1$ is higher than that when $\alpha_{b_i}<1$ given all other parameters remain constant.  We have thus shown that $\alpha_{b_i}^*=1$ whenever $\rho_i>0$.  
\vspace{-.25cm}
\bibliographystyle{ieeetran}
\bibliography{IEEEabrv,twireless}

\begin{thebibliography}{10}
\providecommand{\url}[1]{#1}
\csname url@samestyle\endcsname
\providecommand{\newblock}{\relax}
\providecommand{\bibinfo}[2]{#2}
\providecommand{\BIBentrySTDinterwordspacing}{\spaceskip=0pt\relax}
\providecommand{\BIBentryALTinterwordstretchfactor}{4}
\providecommand{\BIBentryALTinterwordspacing}{\spaceskip=\fontdimen2\font plus
\BIBentryALTinterwordstretchfactor\fontdimen3\font minus
  \fontdimen4\font\relax}
\providecommand{\BIBforeignlanguage}[2]{{%
\expandafter\ifx\csname l@#1\endcsname\relax
\typeout{** WARNING: IEEEtran.bst: No hyphenation pattern has been}%
\typeout{** loaded for the language `#1'. Using the pattern for}%
\typeout{** the default language instead.}%
\else
\language=\csname l@#1\endcsname
\fi
#2}}
\providecommand{\BIBdecl}{\relax}
\BIBdecl

\bibitem{Globecom}
R.~V. Bhat, M.~Motani, and T.~J. Lim, ``{Dual-Path} architecture for energy
  harvesting transmitters with battery discharge constraints,'' in \emph{2015
  IEEE Global Commun. Conf.}, San Diego, USA, Dec. 2015.

\bibitem{Kansal}
A.~Kansal \emph{et~al.}, ``Power management in energy harvesting sensor
  networks,'' \emph{ACM Trans. Embed. Comput. Syst.}, vol.~6, no.~4, Sep. 2007.

\bibitem{non_linear}
M.~Gorlatova, A.~Wallwater, and G.~Zussman, ``Networking low-power energy
  harvesting devices: Measurements and algorithms,'' \emph{IEEE Trans. Mobile
  Comput.}, vol.~12, no.~9, pp. 1853--1865, Sep. 2013.

\bibitem{ieee_ragone}
M.~Mellincovsky \emph{et~al.}, ``Performance and limitations of a constant
  power-fed supercapacitor,'' \emph{IEEE Trans. Energy Convers.}, vol.~29,
  no.~2, pp. 445--452, June 2014.

\bibitem{Krieger}
E.~M. Krieger and C.~B. Arnold, ``Effects of undercharge and internal loss on
  the rate dependence of battery charge storage efficiency,'' \emph{Journal of
  Power Sources}, vol. 210, pp. 286 -- 291, 2012.

\bibitem{Ragone}
T.~Christen and M.~W. Carlen, ``Theory of ragone plots,'' \emph{Journal of
  Power Sources}, vol.~91, no.~2, pp. 210 -- 216, 2000.

\bibitem{Varta}
``Rechargeable button cells: Sales program and technical handbook,''
  \url{http://www.varta-microbattery.com/applications/mb_data/documents/sales_literature_varta/HANDBOOK_Rechargeable_NiMH_Button_en.pdf},
  accessed: 2016-08-08.

\bibitem{Maxwell}
``Datasheet: Hc series ultracapacitors,''
  \url{http://www.maxwell.com/images/documents/hcseries_ds_1013793-9.pdf},
  accessed: 2016-08-08.

\bibitem{EDLC}
``Panasonic:electric double layer capacitors (gold capacitor)/ er,''
  \url{http://www.mouser.com/ds/2/315/Capacitor%20SMT%20Gold%20Cap%20(EEC-ER)-197294.pdf},
  accessed: 2016-08-08.

\bibitem{data}
B.~Varan and A.~Yener, ``Energy harvesting communications with energy and data
  storage limitations,'' in \emph{IEEE Global Commun. Conf.}, Dec. 2014, pp.
  1442--1447.

\bibitem{Review}
S.~Ulukus \emph{et~al.}, ``Energy harvesting wireless communications: A review
  of recent advances,'' \emph{Selected Areas in Communications, IEEE Journal
  on}, vol.~33, no.~3, pp. 360--381, March 2015.

\bibitem{survey}
D.~Gunduz \emph{et~al.}, ``Designing intelligent energy harvesting
  communication systems,'' \emph{IEEE Commun. Mag.}, vol.~52, no.~1, pp.
  210--216, Jan. 2014.

\bibitem{Ozel}
O.~Ozel and S.~Ulukus, ``Achieving awgn capacity under stochastic energy
  harvesting,'' \emph{IEEE Trans. Inform. Theory}, vol.~58, no.~10, pp.
  6471--6483, Oct. 2012.

\bibitem{vsharma_capacity}
R.~Rajesh, V.~Sharma, and P.~Viswanath, ``Capacity of gaussian channels with
  energy harvesting and processing cost,'' \emph{IEEE Transactions on
  Information Theory}, vol.~60, no.~5, pp. 2563--2575, May 2014.

\bibitem{Tutu}
K.~Tutuncuoglu and A.~Yener, ``Optimum transmission policies for battery
  limited energy harvesting nodes,'' \emph{IEEE Trans. Wireless Commun.},
  vol.~11, no.~3, pp. 1180--1189, March 2012.

\bibitem{Jog}
V.~Jog and V.~Anantharam, ``An energy harvesting awgn channel with a finite
  battery,'' in \emph{IEEE Int. Symp. Inform. Theory}, Jun. 2014, pp. 806--810.

\bibitem{Dong}
F.~Amirnavaei and M.~Dong, ``Online power control optimization for wireless
  transmission with energy harvesting and storage,'' \emph{IEEE Trans. Wireless
  Commun.}, vol.~PP, no.~99, pp. 1--1, 2016.

\bibitem{Biplab}
S.~Zhang, A.~Seyedi, and B.~Sikdar, ``An analytical approach to the design of
  energy harvesting wireless sensor nodes,'' \emph{IEEE Trans. Wireless
  Commun.}, vol.~12, no.~8, pp. 4010--4024, August 2013.

\bibitem{leakage}
B.~Devillers and D.~Gunduz, ``A general framework for the optimization of
  energy harvesting communication systems with battery imperfections,''
  \emph{J. Commun. \& Networks}, vol.~14, no.~2, pp. 130--139, Apr. 2012.

\bibitem{Leakage2}
N.~Su and M.~Koca, ``Stochastic transmission policies for energy harvesting
  nodes with random energy leakage,'' in \emph{European Wireless 2014; 20th
  European Wireless Conf.; Proceedings of}, May 2014, pp. 1--6.

\bibitem{shixin}
S.~Luo, R.~Zhang, and T.~J. Lim, ``Optimal save-then-transmit protocol for
  energy harvesting wireless transmitters,'' \emph{IEEE Trans. Wireless
  Commun.}, vol.~12, no.~3, pp. 1196--1207, Mar. 2013.

\bibitem{inefficiency}
K.~Tutuncuoglu \emph{et~al.}, ``Optimum policies for an energy harvesting
  transmitter under energy storage losses,'' \emph{IEEE J. Sel. Areas Commun.},
  vol.~33, no.~3, pp. 467--481, March 2015.

\bibitem{DWF}
O.~Ozel \emph{et~al.}, ``Transmission with energy harvesting nodes in fading
  wireless channels: Optimal policies,'' \emph{IEEE J. Sel. Areas Commun.},
  vol.~29, no.~8, pp. 1732--1743, Sep. 2011.

\bibitem{Rui_DWF}
C.~K. Ho and R.~Zhang, ``Optimal energy allocation for wireless communications
  with energy harvesting constraints,'' \emph{Signal Processing, IEEE Trans.
  on}, vol.~60, no.~9, pp. 4808--4818, Sept 2012.

\bibitem{chandraiisc}
S.~Reddy and C.~R. Murthy, ``Dual-stage power management algorithms for energy
  harvesting sensors,'' \emph{IEEE Transactions on Wireless Communications},
  vol.~11, no.~4, pp. 1434--1445, April 2012.

\bibitem{DGP}
O.~Orhan, D.~Gündüz, and E.~Erkip, ``Energy harvesting broadband
  communication systems with processing energy cost,'' \emph{IEEE Trans.
  Wireless Commun.}, vol.~13, no.~11, pp. 6095--6107, Nov 2014.

\bibitem{Rui_TI}
J.~Xu and R.~Zhang, ``Throughput optimal policies for energy harvesting
  wireless transmitters with non-ideal circuit power,'' \emph{IEEE J Select.
  Areas in Commun.}, vol.~32, no.~2, pp. 322--332, February 2014.

\bibitem{Rui_CET}
X.~Wang and R.~Zhang, ``Optimal transmission policies for energy harvesting
  node with non-ideal circuit power,'' in \emph{Sensing, Communication, and
  Networking (SECON), 2014 Eleventh Annual IEEE Int. Conf. on}, June 2014, pp.
  591--599.

\bibitem{handbook}
D.~Linden and T.~B. Reddy, \emph{Handbook of batteries 4th Ed.}\hskip 1em plus
  0.5em minus 0.4em\relax McGraw Hill, 2011.

\bibitem{TI}
M.~Jensen, ``Coin cells and peak current draw)/ er,''
  \url{http://www.ti.com/lit/wp/swra349/swra349.pdf}, accessed: 2016-08-08.

\bibitem{Kansal_eff_Dual}
V.~Raghunathan \emph{et~al.}, ``Design considerations for solar energy
  harvesting wireless embedded systems,'' in \emph{Information Processing in
  Sensor Networks, 2005. IPSN 2005. Fourth International Symposium on}, April
  2005, pp. 457--462.

\bibitem{Cover06}
T.~M. Cover and J.~A. Thomas, \emph{Elements of Information Theory}.\hskip 1em
  plus 0.5em minus 0.4em\relax Wiley-Interscience, 2006.

\bibitem{complexity}
\BIBentryALTinterwordspacing
F.~A. Potra and S.~J. Wright, ``Interior-point methods,'' \emph{J. Comput.
  Appl. Math.}, vol. 124, no. 1-2, pp. 281--302, Dec. 2000. [Online].
  Available: \url{http://dx.doi.org/10.1016/S0377-0427(00)00433-7}
\BIBentrySTDinterwordspacing

\bibitem{dp}
D.~P. Bertsekas, \emph{Dynamic Programming and Optimal Control}, 2nd~ed.\hskip
  1em plus 0.5em minus 0.4em\relax Athena Scientific, 2000.

\end{thebibliography}
\end{document}